\documentclass[11pt]{article}
\usepackage{graphicx}%插图
\usepackage{float}%浮动
\usepackage{subfigure}%插入多图
\usepackage{caption}
\usepackage{multirow}
\usepackage{makecell}
\usepackage{appendix}
\usepackage[figuresright]{rotating}
\usepackage{booktabs}
\usepackage[linesnumbered,ruled]{algorithm2e}
\usepackage[margin=1in]{geometry}
\usepackage{hyperref}
\usepackage{amsfonts}
\usepackage{mathrsfs}
\usepackage{comment}
\usepackage{amsmath}
\usepackage{amssymb}
\usepackage{amsthm}
\usepackage{amscd}
\usepackage{graphicx}
\usepackage{indentfirst}
\usepackage[all]{xy}
\usepackage{titlesec}
\usepackage{enumerate}
\usepackage{bm}
\usepackage{enumitem}
\usepackage{color}
\usepackage{dsfont}
\usepackage{arydshln}
\usepackage{booktabs}
\newtheorem{theorem}{Theorem}[section]
\newtheorem{lemma}[theorem]{Lemma}
\newtheorem{proposition}[theorem]{Proposition}

\newtheorem{definition}[theorem]{Definition}

\newtheorem{remark}[theorem]{Remark}

\newtheorem{example}[theorem]{Example}

\DeclareMathOperator{\rank}{rank}

\begin{document}
\title{ Two variants of Twisted Reed-Solomon Codes\footnote{The research was supported by the National Natural Science Foundation of China under the Grant No. 12441105.}}
\author{Haojie Gu\footnote{Haojie Gu is with the School of Mathematical Sciences, Capital Normal University, Beijing 100048, China. Email: 2200502051@cnu.edu.cn.},
   \and Huiyue Lei\footnote{Huiyue Lei is with the School of Mathematical Sciences, Capital Normal University, Beijing 100048, China. Email:362885063@qq.com. },
	\and Jun Zhang\footnote{Jun Zhang is with the School of Mathematical Sciences, Capital Normal University, Beijing 100048, China. Email: junz@cnu.edu.cn.}
}

\date{}
\maketitle

\begin{abstract}
	%In this paper, we study a  class of column twisted Reed-Solomon(TRS) codes $CTRS(\mathcal{A},\boldsymbol{B},\boldsymbol{\lambda})$ and a class of row-column TRS codes $RCTRS_{\ell,\eta}(\mathcal{A},\boldsymbol{B},\boldsymbol{\lambda})$, motivated by the works of~\cite{beelen2017twisted,bhagat2025row,liu2025column}. Firstly, we establish some conditions for codes $CTRS(\mathcal{A},\boldsymbol{B},\boldsymbol{\lambda})$ to be MDS codes.  We  also present some explicit constructions for the code $CTRS(\mathcal{A},\boldsymbol{B},\boldsymbol{\lambda})$ being MDS and non-RS MDS codes. Secondly, we establish some conditions for codes $RCTRS_{\ell,\eta}(\mathcal{A},\boldsymbol{B},\boldsymbol{\lambda})$ to be MDS codes and show that the dimension of their Schur square codes are at least $2k+2$. Consequently, these codes are not only not equivalent to RS codes, but also different from the RCTRS codes constructed in~\cite{bhagat2025row,liu2025column}, showing the novelty of our construction.   Finally, we present the dual codes of column TRS codes $CTRS(\mathcal{A},\boldsymbol{B},\boldsymbol{\lambda})$ and row-column TRS codes $RCTRS(\mathcal{A},\boldsymbol{B},\boldsymbol{\lambda})$. 

    Generalized Reed–Solomon codes and twisted generalized Reed–Solomon codes provide important sources of maximum distance separable codes. In this paper, we study two variants obtained by introducing column twists and simultaneous row-column twists into Reed–Solomon type evaluation codes. For the column-twisted family, we provide necessary and sufficient conditions
     for the code to be MDS in terms of explicit subset product conditions. Under the stated parameter assumptions, the Schur square has dimension $2k+1$, 
      which leads to MDS codes that are not equivalent to Reed–Solomon codes. For the row-column twisted family, we establish necessary and sufficient conditions for the MDS property in terms of elementary symmetric functions. The larger Schur square dimension provides a further distinction from both Reed–Solomon codes and known twisted families, thereby yielding new non-RS MDS codes. Finally, explicit parity-check matrices and dual descriptions are obtained for both code families. These results provide a foundation for subsequent studies of self-orthogonality, hull dimensions and applications to quantum-code constructions.
	
	\begin{flushleft}
	\end{flushleft}
\end{abstract}

\section{Introduction}

Let $\mathbb{F}_{q}$ be the  finite field with $q$ elements and $\mathbb{F}_{q}^{*}=\mathbb{F}_{q}\backslash\{0\}$, where $q$ is a power of the prime $p$. Let $\mathbb{F}_{q}^n$ be the $n$-dimensional vector space over the finite field $\mathbb{F}_{q}$. For any vector $ \boldsymbol{x}=(x_1,x_2,\cdots,x_n)\in \mathbb{F}_{q}^n$, the \emph{Hamming weight} $wt( \boldsymbol{x})$ of $ \boldsymbol{x}$ is defined to be the number of non-zero coordinates, i.e.,
$$wt( \boldsymbol{x})=|\left\{i\,|\,1\leqslant i\leqslant n,\,x_i\neq 0\right\}|.$$

An $[n,k,d]$-linear code $\mathcal{C}\subseteq\mathbb{F}_{q}^n$ is a $k$-dimensional linear subspace of $\mathbb{F}_{q}^n$ whose minimum distance $d=d(\mathcal{C})$ is given by
$$d(\mathcal{C})=\min\left\{wt(\boldsymbol{c}):\boldsymbol{c}\in\mathcal{C}\backslash\{0\}\right\}.$$
The well-known Singleton bound states that $d\leq n-k+1$ for any linear code $\mathcal{C}=[n,k,d]$.
A code attaining equality, i.e., with $d=n-k+1$, is called a maximum distance separable (MDS) code.
MDS codes have been extensively studied in the literature, including investigations of their covering radius and deep hole problem~\cite{bartoli2014covering,ZWK2020}, weight distribution~\cite{ezerman2010weights}, LCD property~\cite{carlet2018euclidean}, self-dual property~\cite{ezerman2010weights,lebed2019construction}, classification~\cite{kokkala2015classification} and related cryptographic topics~\cite{cramer2008codes}, etc. 

\begin{definition}
    Let $\mathcal{A}=\left\{\alpha_{1},\cdots,\alpha_{n}\right\}\subseteq\mathbb{F}_{q}$ be the evaluation set and $\boldsymbol{v}=(v_{1},\cdots,v_{n})\in (\mathbb{F}_{q}^{*})^n$, then the generalized Reed-Solomon code $GRS_{k}(\mathcal{A},\boldsymbol{v})$ of length $n$ and dimension $k$ is defined as 
    \begin{equation*}
        GRS_{k}(\mathcal{A},\boldsymbol{v})=\left\{(v_{1}f(\alpha_{1}),\cdots,v_{n}f(\alpha_{n})):f(x)\in\mathbb{F}_{q}[x]_{<k}\right\},
    \end{equation*}
    where $\mathbb{F}_{q}[x]_{<k}:=\left\{f(x)\in\mathbb{F}_{q}[x]:\deg(f(x))<k \right\}$. If $\boldsymbol{v}$ is the all-one vector, it is called a Reed-Solomon code of length $n$ and dimension $k$, denoted as $RS_{k}(\mathcal{A}).$
\end{definition}

Generalized Reed-Solomon $(\operatorname{GRS})$ codes are one of the most important families of MDS codes, as they can correct burst errors and provide high fidelity in CD players. Most known MDS codes are equivalent to GRS  or extended GRS codes. Finding new MDS codes that are not equivalent to GRS codes, which are called non-RS MDS codes, is a challenging problem. This problem has attracted considerable recent attention. It is a well-established fact that any $[n,k]$-MDS code is equivalent to a Reed-Solomon code when $k<3$ or $n-k<3$ \cite{beelen2017twisted}. Therefore, constructing non-RS $[n,k]$-MDS codes is restricted to $3\leq k\leq n-3$. In addition, since the dual of an RS code is itself an RS code, it suffices to focus on the case where $3\leq k\leq n/2$ in the subsequent discussion.

An MDS code that is not equivalent to a generalized Reed-Solomon $(\operatorname{GRS})$ code is called a non-RS MDS code. The first construction of such codes was given by Roth and Lempel in 1989 \cite{roth1989construction}, based on generator matrices and special subsets of finite fields. More recently, in 2016, Sheekey \cite{Sheekey2016anew} introduced a new class of maximum rank distance codes, known as Twisted Gabidulin codes, which are MDS with respect to the rank metric and were shown to be inequivalent to Gabidulin codes (the rank metric analog of Reed-Solomon codes). Inspired by Sheekey's work, Beelen et al. \cite{beelen2017twisted,beelen2022twisted} proposed Twisted Reed-Solomon $(\operatorname{TRS})$ codes and demonstrated that certain families of TRS codes are non-RS MDS codes. Following this development, extensive research has been devoted to the structure and properties of TRS codes \cite{cheng2023parity,ding2025new,gu2023twisted,huang2021mds,sui2022mds1,sui2023new,sui2022mds2,zhang2022class,zhu2024class}. Furthermore, many families of non-RS MDS codes have been constructed using various techniques; see, for instance, \cite{abdukhalikov2025some,chen2023many,jin2025new,li2025non,wu2024more,zhi2025new}.

Recently, Liu et al.~\cite{liu2025column} introduced column-twisted Reed-Solomon $(\operatorname{CTRS})$ codes by twisting a column of the generator matrix of a Reed-Solomon code. They established MDS conditions for these codes and, via Schur square analysis, proved their non-RS MDS property. Subsequently, Bhagat et al.~\cite{bhagat2025row} extended this idea to both rows and columns, proposing Row-Column Twisted Reed-Solomon $(\operatorname{RCTRS})$ codes. They established conditions for these codes to be MDS and demonstrated their existence. In addition, by studying the Schur square of the $\operatorname{RCTRS}$ code, they proved that these codes are non-RS MDS codes and non-equivalent to $\operatorname{CTRS}$ codes.

Compared with the $\operatorname{CTRS}$ and $\operatorname{RCTRS}$ constructions~\cite{bhagat2025row,liu2025column}, the present paper focuses on two special but flexible variants
whose MDS property can be characterized by explicit subset conditions. The main contributions are as follows. First, we provide necessary and sufficient conditions for
 the column-twisted and row-column twisted families to be MDS. Second, we give field-extension and subgroup-based parameter choices that satisfy these conditions and yield explicit non-RS MDS codes. Third, we compute the Schur-square dimensions of the constructed codes.
For the column-twisted family, the Schur square has dimension $2k+1$; for the row-column twisted family, it has dimension $2k+2$ or $2k+3$ in the considered parameter ranges. These dimensions distinguish the constructed codes from GRS codes and, in the corresponding ranges, from previously known twisted families. Finally,
 we determine explicit parity-check matrices and dual descriptions for both families, which may be useful in future studies of self-orthogonality, self-duality, hulls and quantum-code applications.

This paper is organized as follows.
In Section~\ref{sec2}, we present some basic notations and definitions, and introduce two special classes of $\operatorname{CTRS}$ and $\operatorname{RCTRS}$ codes considered in this paper. In Section~\ref{sec3}, we determine a necessary and sufficient condition for  $\operatorname{CTRS}$ codes $\operatorname{CTRS}(\mathcal{A},\boldsymbol{B},\lambda)$ to be MDS. Furthermore, based on this necessary and sufficient condition, we construct MDS codes whose Schur square has dimension $2k+1$, which implies that the resulting codes are non-RS MDS codes. In Section~\ref{sec4}, we determine a necessary and sufficient condition for the $\operatorname{RCTRS}$ codes $\operatorname{RCTRS}_{\ell,\eta}(\mathcal{A},\boldsymbol{B},\lambda)$ to be MDS. Then we provide some constructions of non-RS MDS codes based on $\operatorname{RCTRS}$ codes $\operatorname{RCTRS}_{\ell,\eta}(\mathcal{A},\boldsymbol{B},\boldsymbol{\lambda})$. In Section~\ref{sec5}, we determine the parity-check matrices for $\operatorname{CTRS}$ codes and $\operatorname{RCTRS}$ codes, thus obtaining their dual codes. In Section~\ref{sec6}, we conclude the paper.

\section{Preliminaries}\label{sec2}
In this section, we establish the notation for the rest of the paper. By “natural numbers”, we mean  positive integers, i.e., $\mathbb{N}^{+}=\left\{1,2,3,\cdots\right\}$.  The set of non-negative integers is denoted by $\mathbb{N}$. For $m\in\mathbb{N}^{+}$, let $[m]$ denote the set of integers from $1$ to $m$ and let $[0,m]$ denote the set of integers from $0$ to $m$, i.e., $[m]:=\left\{1,2,\cdots,m\right\}$ and $[0,m]:=\left\{0,1,\cdots,m\right\}$. Given  $n$-subset $U=\left\{\beta_{1},\beta_{2},\cdots,\beta_{n}\right\}\subseteq \mathbb{F}_{q}$ and $k$-subset $I=\left\{i_{1},\cdots,i_{k}\right\}\subseteq [n]$, let $U_{I}=\left\{\beta_{i_{1}},\beta_{i_{2}},\cdots,\beta_{i_{k}}\right\}$ and $V(U_{I}):=V(\beta_{i_{1}},\cdots,\beta_{i_{k}})=\prod\limits_{1\leq j_{1}<j_{2}\leq k}(\beta_{i_{j_{2}}}-\beta_{i_{j_{1}}})$.

 TGRS codes are generalizations of GRS codes and were first introduced in~\cite{beelen2017twisted}.

\begin{definition}[\cite{beelen2017twisted}]
For two positive integers $l, k$ with $l \le k \le n \le q$, suppose that $\mathbf{h}=(h_1,h_2,\ldots,h_l),\mathbf{t}=(t_1,t_2,\ldots,t_l)$ and 
$\boldsymbol{\eta}=(\eta_1,\eta_2,\ldots,\eta_l)\in \mathbb{F}_q^l$, where $0 \le h_i \le k-1$ are distinct and $0 \le t_i \le n-k$ are also distinct.
Then
\[
\mathcal{S}
=
\left\{
\sum_{i=0}^{k-1} f_i x^i
+
\sum_{j=1}^{l} \eta_j f_{h_j} x^{k-1+t_j}
\;:\;
f_0,f_1,\ldots,f_{k-1}\in\mathbb{F}_q
\right\}
\]
is a $k$-dimensional subspace of $\mathbb{F}_q[x]$ over $\mathbb{F}_q$. Furthermore, let
$
\boldsymbol{\alpha}=(\alpha_1,\alpha_2,\ldots,\alpha_n)\in \mathbb{F}_q^n,
$
where $\alpha_i$, $i=1,2,\ldots,n$, are distinct, and
$
\boldsymbol{v}=(v_1,v_2,\ldots,v_n)\in (\mathbb{F}_q^{*})^n.
$
The linear code
\[
\mathcal{C}
=
\left\{
\operatorname{ev}_{\boldsymbol{\alpha},\boldsymbol{v}}(f(x))
\,:\,
f(x)\in\mathcal{S}
\right\}
\]
is called a twisted generalized Reed-Solomon $(\operatorname{TGRS})$ code.
When $\boldsymbol{v}=(1,\ldots,1)$, the code is called 
a twisted Reed-Solomon $(\operatorname{TRS})$ code.
\end{definition}
Later, Hu et al.~\cite{hu2025p} proposed a more general form of TGRS codes.
\begin{definition}[\cite{hu2025p}]
Let $n$, $k$ and $s$ be integers with $0<k\le n$ and $0\le s \le n-k$.
For given
\begin{itemize}
    \item $\mathcal{L}\subseteq [0,n-k-1]$ (called the twist set), where $s:=|\mathcal{L}|$ denotes the number of twists;
    \item $\mathcal{P}\subseteq [0,k-1]$ (called the position set);
    \item $B=[b_{i,j}]\in \mathbb{F}_q^{k\times (n-k)}$ (called the coefficient matrix), where $0\le i\le k-1$ and $0\le j\le n-k-1$,
\end{itemize}
denote by
\begin{equation}\label{Equ:F,L,P,B}
F(\mathcal{L},\mathcal{P},B)
=
\left\{
\sum_{i=0}^{k-1} f_i x^i
+
\sum_{i\in\mathcal{P}} f_i \sum_{j\in\mathcal{L}} b_{i,j}x^{k+j}
:
f_i\in \mathbb{F}_q,\ 0\le i\le k-1
\right\}.
\end{equation}
Let $\boldsymbol{\alpha}=(\alpha_1,\alpha_2,\ldots,\alpha_n)\in \mathbb{F}_q^n$ with pairwise distinct $\alpha_1,\ldots,\alpha_n$, and
$\boldsymbol{v}=(v_1,\ldots,v_n)\in (\mathbb{F}_q^*)^n$.
Then the following twisted generalized Reed-Solomon $(\operatorname{TGRS})$ code 
\begin{equation}
C(\mathcal{L},\mathcal{P},B)
=
\left\{
\operatorname{ev}_{\boldsymbol{\alpha},\boldsymbol{v}}(f(x))
=
(v_1f(\alpha_1),\ldots,v_nf(\alpha_n))
:
f(x)\in F(\mathcal{L},\mathcal{P},B)
\right\}.
\end{equation}
is called an $(\mathcal{L},\mathcal{P})$-TGRS code. When $\boldsymbol{v}=(1,\ldots,1)$, the code is referred to as $(\mathcal{L},\mathcal{P})$-TRS code.
\end{definition}

In this paper, we will consider the special case $\mathcal{L}=\{0\}$ and $\mathcal{P}=\{\ell\}$, where $0\leq\ell\leq k-1\leq n-2$. For any $\eta\in\mathbb{F}_{q}$ denote by
\begin{equation*}
    S_{k,\ell,\eta}=\left\{\sum\limits_{0\leq i\leq k-1\atop i\neq\ell}f_{i}x^i+f_{\ell}\left(x^{\ell}+\eta x^k\right):f_{0},f_{1},\cdots,f_{k-1}\in\mathbb{F}_{q}\right\}.
\end{equation*}
For any $\mathcal{A}=\left\{\alpha_{1},\cdots,\alpha_{n}\right\}\subseteq\mathbb{F}_{q}$, we will focus on the following TRS codes:
\begin{equation*}
    TRS_{k}(\mathcal{A},\ell,\eta)=\left\{(f(\alpha_{1}),\cdots,f(\alpha_{n})):f\in S_{k,\ell,\eta}\right\}
\end{equation*}

 In addition, Liu et al.~\cite{liu2025column} and Bhagat et al.~\cite{bhagat2025row} generalized the way of twisting TRS codes from rows to columns and row-columns, respectively. 
Motivated by their work, we consider the following two classes of variants of TRS codes.
\begin{definition}\label{def:RCTRS}
	Let $0\leq\ell<k\leq n\leq q+2$ and $\mathcal{A}=\left\{\alpha_{1},\cdots,\alpha_{n-2}\right\}\subseteq \mathbb{F}_{q}$, where $\alpha_{1},\cdots,\alpha_{n-2}$ are distinct. Suppose that $\eta\in \mathbb{F}_{q}$, $\boldsymbol{B}=(b_{1},b_{2},b_{3})\in \mathbb{F}_{q}^3$ and $\boldsymbol{\lambda}=(\lambda_{1},\lambda_{2})\in \mathbb{F}_{q}^2$. A column-twisted Reed-Solomon $(\operatorname{CTRS})$ code is defined as 
    \begin{equation*}
    \begin{aligned}
\operatorname{CTRS}(\mathcal{A},\boldsymbol{B},\boldsymbol{\lambda})=&\left\{(f(\alpha_{1}),\cdots,f(\alpha_{n-2}),f(b_{1})-\lambda_{1}f(b_{3}),f(b_{2})-\lambda_{2} f(b_{3})):f(x)\in \mathbb{F}_{q}[x]_{<k}\right\}.
    \end{aligned}
    \end{equation*}
	 A row-column twisted Reed-Solomon $(\operatorname{RCTRS})$ code is defined as \begin{equation*}
     \begin{aligned}
         \operatorname{RCTRS}_{\ell,\eta}(\mathcal{A},\boldsymbol{B},\boldsymbol{\lambda})=&\left\{(f(\alpha_{1}),\cdots,f(\alpha_{n-2}),f(b_{1})-\lambda_{1}f(b_{3}),f(b_{2})-\lambda_{2} f(b_{3})):f(x)\in S_{k,\ell,\eta} \right\}.
     \end{aligned}
\end{equation*}
\end{definition}

Obviously, the corresponding generator matrices are given by 
\begin{equation}
G_{\operatorname{CTRS}}=\begin{pmatrix}
 1&\cdots&1&1-\lambda_{1}&1-\lambda_{2}\\
 \alpha_{1}&\cdots&\alpha_{n-2}&b_{1}-\lambda_{1} b_{3}&b_{2}-\lambda_{2} b_{3}\\
 \vdots&\cdots&\vdots&\vdots&\vdots\\
 \alpha_{1}^{k-1}&\cdots&\alpha_{n-2}^{k-1}&b_{1}^{k-1}-\lambda_{1} b_{3}^{k-1}&b_{2}^{k-1}-\lambda_{2} b_{3}^{k-1}\\
\end{pmatrix}
\end{equation}
and
\small{\begin{equation}
G_{\operatorname{RCTRS}}=\begin{pmatrix}
	1&\cdots&1&1-\lambda_{1}&1-\lambda_{2}\\
	\alpha_{1}&\cdots&\alpha_{n-2}&b_{1}-\lambda_{1}b_{3}&b_{2}-\lambda_{2}b_{3}\\
	\vdots&\vdots&\vdots&\vdots&\vdots\\
	\alpha_{1}^{\ell-1}&\cdots&\alpha_{n-2}^{\ell-1}&b_{1}^{\ell-1}-\lambda_{1}b_{3}^{\ell-1}&b_{2}^{\ell-1}-\lambda_{2}b_{3}^{\ell-1}\\
	\alpha_{1}^{\ell+1}&\cdots&\alpha_{n-2}^{\ell+1}&b_{1}^{\ell+1}-\lambda_{1}b_{3}^{\ell+1}&b_{2}^{\ell+1}-\lambda_{2}b_{3}^{\ell+1}\\
	\vdots&\vdots&\vdots&\vdots&\vdots\\
	\alpha_{1}^{k-1}&\cdots&\alpha_{n-2}^{k-1}&b_{1}^{k-1}-\lambda_{1}b_{3}^{k-1}&b_{2}^{k-1}-\lambda_{2}b_{3}^{k-1}\\
	\alpha_{1}^{\ell}+\eta\alpha_{1}^{k}&\cdots&\alpha_{n-2}^{\ell}+\eta\alpha_{n-2}^k&b_{1}^{\ell}-\lambda_{1}b_{3}^{\ell}+\eta\left( b_{1}^{k}-\lambda_{1}b_{3}^{k}\right)&b_{2}^{\ell}-\lambda_{2}b_{3}^{\ell}+\eta\left( b_{2}^{k}-\lambda_{2}b_{3}^{k}\right)\\
	\end{pmatrix}
\end{equation}}
for the codes $\operatorname{CTRS}(\mathcal{A},\boldsymbol{B},\boldsymbol{\lambda})$ and $\operatorname{RCTRS}_{\ell,\eta}(\mathcal{A},\boldsymbol{B},\boldsymbol{\lambda})$, respectively.

The Schur product of linear codes over a finite field has emerged as a fundamental operation in both classical and quantum coding theory~\cite{randriambololona2015products}. It is widely used in the construction of code-equivalence distinguishers in coding theory and code-based cryptography.
\begin{definition}\label{schur product}
For $\boldsymbol{x}=(x_1,x_2,\cdots ,x_n)$, $\boldsymbol{y}=(y_1,y_2,\cdots ,y_n) \in {\mathbb{F}_q ^n}$, the \textit{Schur product} of $\boldsymbol{x}$ and $\boldsymbol{y}$ is defined as $\boldsymbol{x}*\boldsymbol{y}:=\left( {{x_1}{y_1}, \ldots ,{x_n}{y_n}} \right)$. The Schur product of two linear codes $\mathcal{C}_1,\mathcal{C}_2 \subseteq \mathbb{F}_q^n$ is defined as
\[\mathcal{C}_1*\mathcal{C}_2: =\operatorname{Span}_{\mathbb{F}_{q}}\left\langle {{\textit{\textbf{c}}_1}*{\textit{\textbf{c}}_2}:{\textit{\textbf{c}}_1} \in \mathcal{C}_1,{\textit{\textbf{c}}_2} \in \mathcal{C}_2} \right\rangle.\]
In particular, the Schur product of $\mathcal{C}$ with itself is denoted by $\mathcal{C}^2$ and is often called the Schur square of $\mathcal{C}$.
\end{definition}

The definition of the equivalence for linear codes is given in the following.
\begin{definition}[\cite{beelen2017twisted}]
  Let $\mathcal{C}$,~$\mathcal{D}$ be $[n,k]$ linear codes over $\mathbb{F}_q$. We say that $\mathcal{C}$ and $\mathcal{D}$ are \textit{equivalent} if there is a permutation $\pi \in S_n$ and $\boldsymbol{v}$=$(v_1,v_2,\cdots ,v_n)\in {\left( {\mathbb{F}_q^ * } \right)^n}$ such that $\mathcal{C}={\varphi _{\pi ,\textit{\textbf{v}}}}(\mathcal{D})$ where 
\[{\varphi _{\pi ,v}}:\mathbb{F}_q^n \to \mathbb{F}_q^n,\left( {{c_1}, \ldots ,{c_n}} \right) \mapsto \left( {{v_1}{c_{\pi \left( 1 \right)}}, \ldots ,{v_n}{c_{\pi \left( n \right)}}} \right)\]
is the Hamming-metric isometry of $\mathbb{F}_{q}^n$.
\end{definition}

The following proposition determines the Schur square of a GRS code.

\begin{proposition}[\cite{couvreur2014distinguisher}]\label{distinguisher:RS}
  If $k \leq \frac{n}{2}$, then
$GRS_{k}^{2}(\mathcal{A},\boldsymbol{v}) = GRS_{2k-1}(\mathcal{A},\boldsymbol{v}^2)$. Thus, for $k \leq \frac{n}{2}$, $\dim(GRS_{k}^{2}(\mathcal{A},\boldsymbol{v}))=2k-1.$

\end{proposition}

\begin{remark}
If two $[n,k]$-codes $\mathcal{C}_{1}$ and $\mathcal{C}_{2}$ over $\mathbb{F}_{q}$ are equivalent, then $\mathcal{C}_{1}^2$ and $\mathcal{C}_{2}^2$ are equivalent. Hence, if an $[n,k]$-MDS code $\mathcal{C}$ with $k\leq\frac{n}{2}$ satisfies $\dim(\mathcal{C}^2)\neq 2k-1$, then it is not equivalent to the GRS code. We call such codes non-RS MDS codes.
\end{remark}

\section{Column twisted Reed-Solomon codes
\texorpdfstring{$\operatorname{CTRS}(\mathcal{A},\boldsymbol{B},\boldsymbol{\lambda})$}{CTRS(A,B,lambda)}}\label{sec3}
In this section, we study the column twisted Reed-Solomon code
$\operatorname{CTRS}(\mathcal{A},\boldsymbol{B},\boldsymbol{\lambda})$ introduced in Definition~\ref{def:RCTRS} and explicitly construct non-RS MDS codes $\operatorname{CTRS}(\mathcal{A},\boldsymbol{B},\boldsymbol{\lambda})$. 

We first give the necessary and sufficient conditions for the code $\operatorname{CTRS}(\mathcal{A},\boldsymbol{B},\boldsymbol{\lambda})$ to be MDS in the following theorem.

\begin{theorem}\label{Theorem1}
	The code $\operatorname{CTRS}(\mathcal{A},\boldsymbol{B},\boldsymbol{\lambda})$ is MDS if and only if the following conditions hold:
	\begin{description}
		\item[(i)] for any $k-1$-subset $\mathcal{I}\subseteq [n-2]$, we have 
		\begin{equation*}
		\prod\limits_{i\in \mathcal{I}}(b_{t}-\alpha_{i})\neq\lambda_{t}\cdot 	\prod\limits_{i\in \mathcal{I}}(b_{3}-\alpha_{i}),t=1,2;
		\end{equation*}
		\item[(ii)] for any $k-2$-subset $\mathcal{J}\subseteq [n-2]$, we have
		\begin{equation*}
		\begin{aligned}
		\Phi_{J}(b_{2},b_{1})-\lambda_{2}\Phi_{J}(b_{3},b_{1})+\lambda_{1}\Phi_{J}(b_{3},b_{2})\neq 0,
		\end{aligned}		
		\end{equation*}
		where $\Phi_{J}(x,y)=(x-y)\prod\limits_{j\in \mathcal{J}}(x-\alpha_{j})(y-\alpha_{j})$.
	\end{description}
\end{theorem}

\begin{proof}
	It is well-known that an $[n,k]$-linear code is MDS if and only if any $k$ columns of its generator matrix are linearly independent. Therefore, it is enough to prove that any $k\times k$ submatrix of $G_{\operatorname{CTRS}}$ is invertible. It suffices to show that for any subsets $\mathcal{I}=\left\{i_{1},\cdots,i_{k-1}\right\}$ and $\mathcal{J}=\left\{i_{1},\cdots,i_{k-2}\right\}$ of $[n-2]$, the matrices
	\begin{equation*}
	A_{t}=\begin{pmatrix}
	1&\cdots&1&1-\lambda_{t}\\
	\alpha_{i_{1}}&\cdots&\alpha_{i_{k-1}}&b_{t}-\lambda_{t}b_{3}
   \\
    \vdots&\vdots&\vdots&\vdots\\
    \alpha_{i_{1}}^{k-1}&\cdots&\alpha_{i_{k-1}}^{k-1}&b_{t}^{k-1}-\lambda_{t}b_{3}^{k-1}
    \\
	\end{pmatrix},t=1,2
	\end{equation*}
	and 
	\begin{equation*}
	B=\begin{pmatrix}
	1&\cdots&1&1-\lambda_{1}&1-\lambda_{2}\\
	\alpha_{i_{1}}&\cdots&\alpha_{i_{k-2}}&b_{1}-\lambda_{1}b_{3}
	&b_{2}-\lambda_{2}b_{3}\\
	\vdots&\vdots&\vdots&\vdots&\vdots\\
	\alpha_{i_{1}}^{k-1}&\cdots&\alpha_{i_{k-2}}^{k-1}&b_{1}^{k-1}-\lambda_{1}b_{3}^{k-1}&b_{2}^{k-1}-\lambda_{2}b_{3}^{k-1}
	\\
	\end{pmatrix}
	\end{equation*}
	are invertible. 
    
    (I) We compute the determinant of $A_t$ as follows:
    \begin{equation*}
	\begin{aligned}
	\det(A_{t})&=\left|
	\begin{array}{cccc}
			1&\cdots&1&1\\
			\alpha_{i_{1}}&\cdots&\alpha_{i_{k-1}}&b_{t}
			\\
			\vdots&\vdots&\vdots&\vdots\\
			\alpha_{i_{1}}^{k-1}&\cdots&\alpha_{i_{k-1}}^{k-1}&b_{t}^{k-1}
		\end{array}
	\right|-\lambda_{t}\left|
	\begin{array}{cccc}
	1&\cdots&1&1\\
	\alpha_{i_{1}}&\cdots&\alpha_{i_{k-1}}&b_{3}
	\\
	\vdots&\vdots&\vdots&\vdots\\
	\alpha_{i_{1}}^{k-1}&\cdots&\alpha_{i_{k-1}}^{k-1}&b_{3}^{k-1}
	\end{array}
	\right|\\
	&=\left(\prod\limits_{1\leq j\leq k-1}(b_{t}-\alpha_{i_{j}})-\lambda_{t}\prod\limits_{1\leq j\leq k-1}(b_{3}-\alpha_{i_{j}})\right)\prod\limits_{1\leq j_{1}<j_{2}\leq k-1}(\alpha_{i_{j_{2}}}-\alpha_{i_{j_{1}}})
	\end{aligned}
	\end{equation*}
	Thus, $\det(A_{t})\neq 0$ if and only if $\prod\limits_{1\leq j\leq k-1}(b_{t}-\alpha_{i_{j}})\neq \lambda_{t}\prod\limits_{1\leq j\leq k-1}(b_{3}-\alpha_{i_{j}})$.
    
    (II) We compute the determinant of $B$ as follows:
    \begin{equation*}
	\begin{aligned}
	\det(B)&=\left|
	\begin{array}{ccccc}
	1&\cdots&1&1&1\\
	\alpha_{i_{1}}&\cdots&\alpha_{i_{k-2}}&b_{1}&b_{2}
	\\
	\vdots&\vdots&\vdots&\vdots&\vdots\\
	\alpha_{i_{1}}^{k-1}&\cdots&\alpha_{i_{k-2}}^{k-1}&b_{1}^{k-1}&b_{2}^{k-1}
	\end{array}
	\right|-\lambda_{2}\left|
	\begin{array}{ccccc}
	1&\cdots&1&1&1\\
	\alpha_{i_{1}}&\cdots&\alpha_{i_{k-2}}&b_{1}&b_{3}
	\\
	\vdots&\vdots&\vdots&\vdots&\vdots\\
	\alpha_{i_{1}}^{k-1}&\cdots&\alpha_{i_{k-2}}^{k-1}&b_{1}^{k-1}&b_{3}^{k-1}
	\end{array}
	\right|\\
	&+\lambda_{1}\left|
	\begin{array}{ccccc}
	1&\cdots&1&1&1\\
	\alpha_{i_{1}}&\cdots&\alpha_{i_{k-2}}&b_{2}&b_{3}
	\\
	\vdots&\vdots&\vdots&\vdots&\vdots\\
	\alpha_{i_{1}}^{k-1}&\cdots&\alpha_{i_{k-2}}^{k-1}&b_{2}^{k-1}&b_{3}^{k-1}
	\end{array}
	\right|\\
	&=\left(  \Phi_{J}(b_{2},b_{1})-\lambda_{2}\Phi_{J}(b_{3},b_{1})+\lambda_{1}\Phi_{J}(b_{3},b_{2})   \right)\cdot \prod\limits_{1\leq j_{1}<j_{2}\leq k-2}(\alpha_{i_{j_{2}}}-\alpha_{i_{j_{1}}}),
	\end{aligned}
	\end{equation*}
	where $\Phi_{J}(x,y)=(x-y)\prod\limits_{j\in \mathcal{J}}(x-\alpha_{j})(y-\alpha_{j})$. Thus, $\det(B)\neq 0$ if and only if $\Phi_{J}(b_{2},b_{1})-\lambda_{2}\Phi_{J}(b_{3},b_{1})+\lambda_{1}\Phi_{J}(b_{3},b_{2})\neq 0$. This completes the proof.
\end{proof}

The following example gives an MDS code $\operatorname{CTRS}(\mathcal{A},\boldsymbol{B},\boldsymbol{\lambda})$.
\begin{example}
    Let $n=11,k=5$ and $\mathcal{A}=\{1,2,3,5,6,7,9,12,15\}\subseteq\mathbb{F}_{17}$. Choose $b_{1}=11,b_2=b_3=14,\lambda_{1}=0,\lambda_{2}=9$. 
    A direct computation verifies that $\prod\limits_{i\in \mathcal{I}}(b_{t}-\alpha_{i})\neq\lambda_{t}\cdot 	\prod\limits_{i\in \mathcal{I}}(b_{3}-\alpha_{i}),t=1,2$ for any $4$-subset $\mathcal{I}\subseteq [9]$ and $\Phi_{J}(b_{2},b_{1})-\lambda_{2}\Phi_{J}(b_{3},b_{1})+\lambda_{1}\Phi_{J}(b_{3},b_{2})\neq 0$ for any $3$-subset $\mathcal{J}\subseteq [9]$ by using Magma. So by Theorem~\ref{Theorem1}, the code $\operatorname{CTRS}(\mathcal{A},\boldsymbol{B},\boldsymbol{\lambda})$ is an MDS code with parameters $[11,5,7]$ over $\mathbb{F}_{17}$. %by using Magma.
\end{example}

Next, we present some explicit constructions for the code $\operatorname{CTRS}(\mathcal{A},\boldsymbol{B},\boldsymbol{\lambda})$, which are MDS and non-RS MDS codes.

\begin{theorem}\label{The:3.2}
    Let $\mathbb{F}_{q_{0}}\subsetneq\mathbb{F}_{q_{1}}\subsetneq\mathbb{F}_{q}$ be a chain of finite fields. Let $\mathcal{A}=\left\{\alpha_{1},\cdots,\alpha_{n-2}\right\}\subseteq \mathbb{F}_{q_{0}}$ and $\boldsymbol{B}=\left(b_{1},b_{2},b_{3}\right)\in\mathbb{F}_{q_{0}}^3$ be such that $b_{1},b_{2},b_{3}$ are not all the same and $\alpha_{i}\neq b_{j}$ for all $1\leq i\leq n-2,1\leq j\leq 3$. Let $\boldsymbol{\lambda}=(\lambda_{1},\lambda_{2})\in\mathbb{F}_{q}^2$ be such that $\lambda_{1}\in \mathbb{F}_{q_{1}}^{*}\backslash \mathbb{F}_{q_{0}}^*,\lambda_{2}\in \mathbb{F}_{q}^{*}\backslash \mathbb{F}_{q_{1}}^*$. Then
    \begin{description}
        \item[(i)] The code $\operatorname{CTRS}(\mathcal{A},\boldsymbol{B},\boldsymbol{\lambda})$ is an MDS code.
        \item[(ii)] If $4\leq k\leq\frac{n-1}{2}$ and $b_{1},b_{2},b_{3}$ are pairwise distinct, then the Schur square code of $\operatorname{CTRS}(\mathcal{A},\boldsymbol{B},\boldsymbol{\lambda})$ has dimension $2k+1$. Thus, the code $\operatorname{CTRS}(\mathcal{A},\boldsymbol{B},\boldsymbol{\lambda})$ is a non-RS MDS code.
    \end{description}
\end{theorem}

\begin{proof}
	(i) First, we prove that the code $\operatorname{CTRS}(\mathcal{A},\boldsymbol{B},\boldsymbol{\lambda})$ is MDS. For any $k-1$-subset $I\subseteq [n-2]$, we know
	\begin{equation*}
	\frac{\prod\limits_{i\in I}(b_{t}-\alpha_{i})}{\prod\limits_{i\in I}(b_{3}-\alpha_{i})}\in \mathbb{F}_{q_{0}}^{*}\ \mbox{for}\ t=1,2,\lambda_{1}\in \mathbb{F}_{q_1}^*\backslash \mathbb{F}_{q_{0}}^*\ \mbox{and}\ \lambda_{2}\in \mathbb{F}_{q}^{*}\backslash \mathbb{F}_{q_{1}}^*.
	\end{equation*}
	Thus, for any $k-1$-subset $\mathcal{I}\subseteq [n-2]$, we have\begin{equation*}
	\prod\limits_{i\in \mathcal{I}}(b_{t}-\alpha_{i})\neq\lambda_{t}\cdot\prod\limits_{i\in \mathcal{I}}(b_{3}-\alpha_{i}),t=1,2.
	\end{equation*} 
	Next, suppose that there is a subset $\mathcal{J}=\left\{i_{1},\cdots,i_{k-2}\right\}\subseteq [n-2]$ such that $$\Phi_{J}(b_{2},b_{1})-\lambda_{2}\Phi_{J}(b_{3},b_{1})+\lambda_{1}\Phi_{J}(b_{3},b_2)=0.$$
	If $b_{1}=b_{2}$, then $b_{3}\neq b_{t},\Phi_{J}(b_{3},b_{t})\neq 0,t=1,2$ and $\Phi_{J}(b_{2},b_{1})=0$. Thus,
	\begin{equation*}
	\lambda_{1}\frac{\Phi_{J}(b_{3},b_{2})}{\Phi_{J}(b_{3},b_{1})}\in \mathbb{F}_{q_{1}}^{*},\ \lambda_{2}\in \mathbb{F}_{q}^{*}\backslash \mathbb{F}_{q_{1}}^*,
	\end{equation*}
	which contradicts $\lambda_{2}=\lambda_{1}\frac{\Phi_{J}(b_{3},b_{2})}{\Phi_{J}(b_{3},b_{1})}$. The following two situations are similar: $(1)$\ $b_{1}\neq b_{2},b_{3}=b_{t},t=1,2$; $(2)$\ $b_{1},b_{2},b_{3}$ are pairwise distinct . Thus, for any $k-2$-subset $\mathcal{J}\subseteq [n-2]$, we have
	\begin{equation*}
	\begin{aligned}
	\Phi_{J}(b_{2},b_{1})-\lambda_{2}\Phi_{J}(b_{3},b_{1})+\lambda_{1}\Phi_{J}(b_{3},b_{2})\neq 0,
	\end{aligned}		
	\end{equation*}
	From Theorem~\ref{Theorem1}, we know that $\operatorname{CTRS}(\mathcal{A},\boldsymbol{B},\boldsymbol{\lambda})$ is MDS.
	
	(ii) To prove that $\mathcal{C}:=\operatorname{CTRS}(\mathcal{A},\boldsymbol{B},\boldsymbol{\lambda})$ is a non-RS code, we compute the dimension of the Schur square $\mathcal{C}^2$ of $\mathcal{C}$. By definition, $\mathcal{C}^2$ is generated by the following rows
	\begin{equation*}
	\boldsymbol{g}_{i,j}:=\left(\alpha_{1}^{i+j},\cdots,\alpha_{n-2}^{i+j},(b_{1}^i-\lambda_{1}b_{3}^i)(b_{1}^j-\lambda_{1}b_{3}^j),(b_{2}^i-\lambda_{2}b_{3}^i)(b_{2}^j-\lambda_{2}b_{3}^j)\right),
	\end{equation*}
    for $0\leq i,j\leq k-1$. Choosing $(i,j)=(0,2),(1,1)$ respectively, we obtain two rows in $\mathcal{C}^2$:
	\begin{equation*}
	\boldsymbol{g}_{0,2}=\left(\alpha_{1}^{2},\cdots,\alpha_{n-2}^{2},(1-\lambda_{1})(b_{1}^2-\lambda_{1}b_{3}^2),(1-\lambda_{2})(b_{2}^2-\lambda_{2}b_{3}^2)\right)
	\end{equation*}
	and \begin{equation*}
	\boldsymbol{g}_{1,1}=\left(\alpha_{1}^{2},\cdots,\alpha_{n-2}^{2},(b_{1}-\lambda_{1}b_{3})(b_{1}-\lambda_{1}b_{3}),(b_{2}-\lambda_{2}b_{3})(b_{2}-\lambda_{2}b_{3})\right).
	\end{equation*}
	The difference between the above two rows is 
	\begin{equation*}
	-\boldsymbol{g}_{0,2}+\boldsymbol{g}_{1,1}=(0,\cdots,0,\lambda_{1}(b_{1}-b_{3})^2,\lambda_{2}(b_{2}-b_{3})^2).
	\end{equation*}
	Similarly, Choosing $(i,j)=(0,3),(1,2)$, we have
	\begin{equation*}
	\begin{aligned}
	-\boldsymbol{g}_{0,3}+\boldsymbol{g}_{1,2}&=-\left(\alpha_{1}^{3},\cdots,\alpha_{n-2}^{3},(1-\lambda_{1})(b_{1}^3-\lambda_{1}b_{3}^3),(1-\lambda_{2})(b_{2}^3-\lambda_{2}b_{3}^3)\right)\\
	&+\left(\alpha_{1}^{3},\cdots,\alpha_{n-2}^{3},(b_{1}-\lambda_{1}b_{3})(b_{1}^2-\lambda_{1}b_{3}^2),(b_{2}-\lambda_{2}b_{3})(b_{2}^2-\lambda_{2}b_{3}^2)\right)\\
	&=(0,\cdots,0,\lambda_{1}(b_{1}-b_{3})^2(b_{3}+b_{1}),\lambda_{2}(b_{2}-b_{3})^2(b_{2}+b_{3})).
	\end{aligned}	
	\end{equation*}
	Since $$\left| \begin{array}{cc}
	\lambda_{1}(b_{1}-b_{3})^2&\lambda_{2}(b_{2}-b_{3})^2\\
	\lambda_{1}(b_{1}-b_{3})^2(b_{1}+b_{3})&\lambda_{2}(b_{2}-b_{3})^2(b_{2}+b_{3})
	\end{array}\right|=\lambda_{1}\lambda_{2}(b_{1}-b_{3})^2(b_{2}-b_{3})^2(b_{2}-b_{1})\neq 0,$$ $(0,\cdots,0,\beta_{1},\beta_{2})\in \mathbb{F}_{q}^n$ can be linearly expressed by $\boldsymbol{g}_{1,1}-\boldsymbol{g}_{0,2}$ and $\boldsymbol{g}_{1,2}-\boldsymbol{g}_{0,3}$, for any $\beta_{1},\beta_{2}\in \mathbb{F}_{q}$. Thus, for any $(i_{1},j_{1}),(i_{2},j_{2})\in [0,k-1]\times [0,k-1]$ that satisfies $i_{1}+j_{1}=i_{2}+j_{2}$, we find that $\boldsymbol{g}_{i_1,j_1}-\boldsymbol{g}_{i_2,j_2}$ can be linearly expressed by $\boldsymbol{g}_{1,1}-\boldsymbol{g}_{0,2}$ and $\boldsymbol{g}_{1,2}-\boldsymbol{g}_{0,3}$. 
	
	Therefore, $\mathcal{C}^2$ is generated by the following matrix
	\small{\begin{equation*}
	\begin{pmatrix}
	1&\cdots&1&(1-\lambda_{1})^2&(1-\lambda_{2})^2\\
	\alpha_{1}&\cdots&\alpha_{n-2}&(1-\lambda_{1})(b_1-\lambda_{1}b_{3})&(1-\lambda_{2})(b_2-\lambda_{2}b_{3})\\
	\alpha_{1}^2&\cdots&\alpha_{n-2}^2&(1-\lambda_{1})(b_1^2-\lambda_{1}b_{3}^2)&(1-\lambda_{2})(b_2^2-\lambda_{2}b_{3}^2)\\
	\alpha_{1}^3&\cdots&\alpha_{n-2}^3&(1-\lambda_{1})(b_1^3-\lambda_{1}b_{3}^3)&(1-\lambda_{2})(b_2^3-\lambda_{2}b_{3}^3)\\
	\vdots&\vdots&\vdots&\vdots&\vdots\\
     \alpha_{1}^{k-1}&\cdots&\alpha_{n-2}^{k-1}&(1-\lambda_{1})(b_{1}^{k-1}-\lambda_{1}b_{3}^{k-1}&(1-\lambda_{2})(b_{2}^{k-1}-\lambda_{2}b_{3}^{k-1})\\
     \alpha_{1}^{k}&\cdots&\alpha_{n-2}^{k}&(b_1-\lambda_{1}b_3)(b_{1}^{k-1}-\lambda_{1}b_{3}^{k-1})&(b_2-\lambda_{2}b_3)(b_{2}^{k-1}-\lambda_{2}b_{3}^{k-1})\\
     \vdots&\cdots&\vdots&\vdots&\vdots\\
	\alpha_{1}^{2k-2}&\cdots&\alpha_{n-2}^{2k-2}&(b_{1}^{k-1}-\lambda_{1}b_{3}^{k-1})(b_1^{k-1}-\lambda_{1}b_{3}^{k-1})&(b_{2}^{k-1}-\lambda_{2}b_{3}^{k-1})(b_2^{k-1}-\lambda_{2}b_{3}^{k-1})\\
	\alpha_{1}^2&\cdots&\alpha_{n-2}^2&(b_{1}-\lambda_{1}b_{3})(b_1-\lambda_{1}b_{3})&(b_{2}-\lambda_{2}b_{3})(b_2-\lambda_{2}b_{3})\\
	\alpha_{1}^3&\cdots&\alpha_{n-2}^3&(b_{1}-\lambda_{1}b_{3})(b_1^2-\lambda_{1}b_{3}^2)&(b_{2}-\lambda_{2}b_{3})(b_2^2-\lambda_{2}b_{3}^2)\\
    \alpha_{1}^4&\cdots&\alpha_{n-2}^4&(b_{1}-\lambda_{1}b_{3})(b_1^3-\lambda_{1}b_{3}^3)&(b_{2}-\lambda_{2}b_{3})(b_2^3-\lambda_{2}b_{3}^3)\\
    \alpha_{1}^4&\cdots&\alpha_{n-2}^4&(b_{1}^2-\lambda_{1}b_{3}^2)(b_1^2-\lambda_{1}b_{3}^2)&(b_{2}^2-\lambda_{2}b_{3}^2)(b_2^2-\lambda_{2}b_{3}^2)\\
	\vdots&\vdots&\vdots&\vdots&\vdots\\
    \alpha_{1}^{2k-4}&\cdots&\alpha_{n-2}^{2k-4}&(b_{1}^{k-2}-\lambda_{1}b_{3}^{k-2})(b_1^{k-2}-\lambda_{1}b_{3}^{k-2})&(b_{2}^{k-2}-\lambda_{2}b_{3}^{k-2})(b_2^{k-2}-\lambda_{2}b_{3}^{k-2})\\
	\end{pmatrix},
	\end{equation*}}
	which is row-equivalent to
	\small{\begin{equation}\label{Equ:CTRS,C^2,Schur}
	\begin{pmatrix}
	1&\cdots&1&0&0\\
	\alpha_{1}&\cdots&\alpha_{n-2}&0&0\\
	\alpha_{1}^2&\cdots&\alpha_{n-2}^2&0&0\\
	\alpha_{1}^3&\cdots&\alpha_{n-2}^3&0&0\\
	\vdots&\vdots&\vdots&\vdots&\vdots\\
	\alpha_{1}^{2k-2}&\cdots&\alpha_{n-2}^{2k-2}&0&0\\
	0&\cdots&0&\lambda_{1}(b_{1}-b_{3})^2&\lambda_{2}(b_{2}-b_{3})^2\\
	0&\cdots&0&\lambda_{1}(b_{1}-b_{3})^2(b_{1}+b_{3})&\lambda_{2}(b_{2}-b_{3})^2(b_{2}+b_{3})\\
	\end{pmatrix}.
	\end{equation}}
    Since $2k-1\leq n-2$ and
    \begin{equation*}
    \begin{aligned}
        & \det\begin{pmatrix}
            1&\cdots&1&0&0\\
            \vdots&\vdots&\vdots&\vdots&\vdots\\
            \alpha_1^{2k-2}&\cdots&\alpha_{2k-1}^{2k-2}&0&0\\
            0&\cdots&0&\lambda_{1}(b_{1}-b_{3})^2&\lambda_{2}(b_{2}-b_{3})^2\\
            0&\cdots&0&\lambda_{1}(b_{1}-b_{3})^2(b_{1}+b_{3})&\lambda_{2}(b_{2}-b_{3})^2(b_{2}+b_{3})
        \end{pmatrix}\\
        =&\lambda_{1}\lambda_{2}(b_{1}-b_{3})^2(b_{2}-b_{3})^2(b_{2}-b_{1})\cdot V(\boldsymbol{\alpha}_{[2k-1]})\neq 0,
    \end{aligned}
    \end{equation*}
    we know that the rank of matrix~\eqref{Equ:CTRS,C^2,Schur} is $2k+1$.
	 As a result, the code $\operatorname{CTRS}(\mathcal{A},\boldsymbol{B},\boldsymbol{\lambda})$ has Schur square dimension $2k+1$, and therefore is a non-RS MDS code by Proposition~\ref{distinguisher:RS}.
\end{proof}

\begin{remark}
    Since the dimension of  Schur square of the code $\operatorname{CTRS}(\mathcal{A},\boldsymbol{B},\boldsymbol{\lambda})$ constructed here is $2k+1$, which differs from that of the $\operatorname{RCTRS}$ codes constructed in~\cite{bhagat2025row}, we obtain new non-RS MDS codes distinct from those in~\cite{bhagat2025row}.
\end{remark}

\begin{example}
    Consider the chain of finite fields $\mathbb{F}_{13}\subseteq\mathbb{F}_{13^2}\subseteq\mathbb{F}_{13^4}$. Let $n=12,k=4,\alpha_{i}=i-1\in\mathbb{F}_{13}$ for all $1\leq i\leq 10$ and $b_{1}=10,b_{2}=11,b_{3}=12\in\mathbb{F}_{13}$. Let $\gamma$ be a primitive element of $\mathbb{F}_{13^4}$ and $\lambda_{1}=\gamma^{169},\lambda_{2}=\gamma$. For $\mathcal{A}=\left\{\alpha_{1},\cdots,\alpha_{10}\right\},\boldsymbol{B}=(b_{1},b_{2},b_{3}),\boldsymbol{\lambda}=(\lambda_{1},\lambda_{2})$, the codes $\operatorname{CTRS}(\mathcal{A},\boldsymbol{B},\boldsymbol{\lambda})$ are MDS codes with parameters $[12,4,9]$ over $\mathbb{F}_{13^4}$. In addition, the Schur square code of $\operatorname{CTRS}(\mathcal{A},\boldsymbol{B},\boldsymbol{\lambda})$ has dimension $9\neq 7$. Thus, the codes $\operatorname{CTRS}(\mathcal{A},\boldsymbol{B},\boldsymbol{\lambda})$ are non-RS MDS codes.
\end{example}

Indeed, the multiplicative subgroup $\mathbb{F}_{q_{0}}^*$ used in Theorem~\ref{The:3.2}(ii) can be generalized to any multiplicative subgroup of $\mathbb{F}_{q_{1}}^*$. Taking  proper multiplicative subgroups, one can obtain non-RS  MDS codes $\operatorname{CTRS}(\mathcal{A},\boldsymbol{B},\boldsymbol{\lambda})$ with longer lengths compared to those in Theorem~\ref{The:3.2}(ii).

%Since the cardinality of the largest multiplicative subgroup of a field is larger than that of its maximal subfield, the next theorem gives explicit constructions of non-RS  MDS codes $CTRS(\mathcal{A},\boldsymbol{B},\boldsymbol{\lambda})$ with longer lengths compared to Theorem~\ref{The:3.2}.
%For example, suppose that $\mathbb{F}_{q_0}\subseteq \mathbb F_q$ and $q_0=p^m$  is an odd prime power. Then the code constructed in Theorem~\ref{The:3.3} can have length $\frac{p^m+1}{2}$, while the code constructed in Theorem~\ref{The:3.2} has length at most $p^s+2$, which is much smaller than  $\frac{p^m+1}{2}$, where $s$ is a divisor of $m$.

\begin{theorem}
    \label{The:3.3}
    Let $\mathbb{F}_{q_{1}}$ be a subfield of $\mathbb{F}_{q}$ and $H$ be a subgroup of $\mathbb{F}_{q_{1}}^{*}$ of order $n-1$. Let $\boldsymbol{B}=\left(b_{1},b_{2},b_{3}\right)\in\mathbb{F}_{q_{1}}^3,\boldsymbol{\lambda}=(\lambda_{1},\lambda_{2})$ and $H=\left\{1,\mu_{1},\cdots,\mu_{n-2}\right\}$ with $\mu_i\neq 1$ for all $1\leq i\leq n-2$, where $b_{1},b_{2},b_{3}$ are pairwise distinct, $\lambda_{1}\in\mathbb{F}_{q_{1}}^{*}\backslash H$ and $\lambda_{2}\in\mathbb{F}_{q}^{*}\backslash\mathbb{F}_{q_{1}}^*$. Let $\alpha_{i}=\frac{b_{3}\mu_{i}-b_{1}}{\mu_{i}-1}$ for all $1\leq i\leq n-2$. If $4\leq k\leq\frac{n-1}{2}$, then  the code $\operatorname{CTRS}(\mathcal{A},\boldsymbol{B},\boldsymbol{\lambda})$ is a non-RS MDS code.
   
\end{theorem}

\begin{proof}
	 Since $b_3\mu_i-b_1\neq b_3(\mu_i-1)$ and $b_3\mu_i-b_1\neq b_1(\mu_i-1)$, we know that $\alpha_{i}\neq b_j$ for all $1\leq i\leq n-2$ and $j=1,3$.
    We first prove that the code $\operatorname{CTRS}(\mathcal{A},\boldsymbol{B},\boldsymbol{\lambda})$ is MDS.
For any $k-1$-subset $\mathcal{I}\subseteq [n-2]$, we know that
\[
\frac{\prod\limits_{i\in I}(b_{1}-\alpha_{i})}{\prod\limits_{i\in I}(b_{3}-\alpha_{i})}=\prod\limits_{i\in I}\frac{b_{1}-\alpha_{i}}{b_{3}-\alpha_{i}}=\prod\limits_{i\in I}\mu_{i}\in H,\ \ \lambda_{1}\in \mathbb{F}_{q_{1}}^{*}\backslash H
\]
and 
\[
 \frac{\prod\limits_{i\in I}(b_{2}-\alpha_{i})}{\prod\limits_{i\in I}(b_{3}-\alpha_{i})}\in \mathbb{F}_{q_1},\ \ \lambda_{2}\in\mathbb{F}_{q}^{*}\backslash\mathbb{F}_{q_{1}}^*,
\]
which means 
$$\prod\limits_{i\in\mathcal{I}}(b_{t}-\alpha_{i})\neq\lambda_{t}\prod\limits_{i\in\mathcal{I}}(b_{3}-\alpha_{i})\quad\mbox{for }\quad t=1,2.$$

    Next, suppose that there is a subset $\mathcal{J}=\left\{i_{1},\cdots,i_{k-2}\right\}\subseteq [n-2]$ such that $$\Phi_{J}(b_{2},b_{1})-\lambda_{2}\Phi_{J}(b_{3},b_{1})+\lambda_{1}\Phi_{J}(b_{3},b_2)=0.$$
%If there exists $1\leq i\leq n-2$ such that $\alpha_{i}=b_{2}$, then $\Phi_{\mathcal{J}}(b_2,b_1)=\Phi_{\mathcal{J}}(b_3,b_2)=0$. Thus, $\lambda_{2}\Phi_{\mathcal{J}}(b_3,b_1)=0$, which contradicts $\lambda_{2}\Phi_{\mathcal{J}}(b_3,b_1)\neq 0$. Therefore, let $\alpha_{i}\neq b_2$ for all $1\leq i\leq n-2$. 

Since $b_{1},b_{2},b_{3}$ are pairwise distinct, then $\Phi_{J}(b_{3},b_{1})\neq 0$. In addition, $$\frac{\Phi_{J}(b_{2},b_{1})}{\Phi_{J}(b_{3},b_{1})}+\lambda_{1}\frac{\Phi_{J}(b_{3},b_{2})}{\Phi_{J}(b_{3},b_{1})}\in\mathbb{F}_{q_{1}},\ \ \lambda_{2}\in \mathbb{F}_{q}^{*}\backslash \mathbb{F}_{q_{1}}^*,$$
   which  contradicts $\lambda_{2}=\frac{\Phi_{J}(b_{2},b_{1})}{\Phi_{J}(b_{3},b_{1})}+\lambda_{1}\frac{\Phi_{J}(b_{3},b_{2})}{\Phi_{J}(b_{3},b_{1})}$.
	  Thus, for any $k-2$-subset $\mathcal{J}\subseteq [n-2]$, we have
	\begin{equation*}
	\begin{aligned}
	\Phi_{J}(b_{2},b_{1})-\lambda_{2}\Phi_{J}(b_{3},b_{1})+\lambda_{1}\Phi_{J}(b_{3},b_{2})\neq 0.
	\end{aligned}		
	\end{equation*}
	From Theorem~\ref{Theorem1}, we know that $\operatorname{CTRS}(\mathcal{A},\boldsymbol{B},\boldsymbol{\lambda})$ is MDS. 
    
    For calculating the dimension of the Schur square of $\operatorname{CTRS}(\mathcal{A},\boldsymbol{B},\boldsymbol{\lambda})$, similar to Theorem~\ref{The:3.2}, the Schur square code of $\operatorname{CTRS}(\mathcal{A},\boldsymbol{B},\boldsymbol{\lambda})$ has dimension $2k+1$ for $4\leq k\leq\frac{n-1}{2}$. Thus, the code $\operatorname{CTRS}(\mathcal{A},\boldsymbol{B},\boldsymbol{\lambda})$ is a non-RS MDS code by Proposition~\ref{distinguisher:RS}.
    \end{proof}

\section{Row-column twisted Reed-Solomon codes
\texorpdfstring{$\operatorname{RCTRS}_{\ell,\eta}(\mathcal{A},\boldsymbol{B},\boldsymbol{\lambda})$}{CTRS(A,B,lambda)}}\label{sec4}

In this section, we first study the MDS property of row-column twisted Reed-Solomon codes $\operatorname{RCTRS}_{\ell,\eta}(\mathcal{A},\boldsymbol{B},\boldsymbol{\lambda})$ and then explicitly construct non-RS MDS codes $\operatorname{RCTRS}_{\ell,\eta}(\mathcal{A},\boldsymbol{B},\boldsymbol{\lambda})$. 

The following lemma plays an important role in determining the necessary and sufficient conditions for these codes to be MDS codes.

\begin{lemma}[{\cite[Lemma 2.3]{yan2024mutually}}]\label{Lem:4.1}
	Let $m$ be a fixed positive integer and $$I_{m}=\left\{0,1,\cdots,m-1\right\}=\left\{t_{1},\cdots,t_{s}\right\}\bigcup\left\{r_{1},r_{2},\cdots,r_{s^{'}}\right\}$$ be any partition of $I_{m}$ with $m=s+s^{\prime},0=t_{1}<t_2<\cdots<t_{s}=m-1$ and $r_{1}<r_{2}<\cdots<r_{s^{'}}$.
	For any $\mathcal{S}=\left\{a_{1},a_{2},\cdots,a_{s}\right\}\subseteq \mathbb{F}_{q}$, denote by $S_{i}(\mathcal{S})=\sum\limits_{1\leq j_{1}<\cdots<j_{i}\leq s}\prod\limits_{t=1}^{i}a_{j_{t}}$. Then we have the following determinant formula
	\begin{equation*}
	\det \begin{pmatrix}
	a_{1}^{t_{1}}&a_{2}^{t_{1}}&\cdots&a_{s}^{t_{1}}\\
	a_{1}^{t_{2}}&a_{2}^{t_{2}}&\cdots&a_{s}^{t_{2}}\\
	\vdots&\vdots&\vdots&\vdots\\
	a_{1}^{t_{s}}&a_{2}^{t_{s}}&\cdots&a_{s}^{t_{s}}\\
	\end{pmatrix}=\prod\limits_{1\leq i<j\leq s}(a_{j}-a_{i})\cdot\vartriangle,
	\end{equation*}
	where $\vartriangle$ denotes the determinant of the following matrix
	$$\begin{pmatrix}
	S_{s-r_{1}}(\mathcal{S})&S_{s-r_{2}}(\mathcal{S})&\cdots&S_{s-r_{s^{\prime}}}(\mathcal{S})\\
	S_{s-r_{1}+1}(\mathcal{S})&S_{s-r_{2}+1}(\mathcal{S})&\cdots&S_{s-r_{s^{\prime}}+1}(\mathcal{S})\\
	\vdots&\vdots&\vdots&\vdots\\
	S_{s-r_{1}+s^{\prime}-1}(\mathcal{S})&S_{s-r_{2}+s^{\prime}-1}(\mathcal{S})&\cdots&S_{s-r_{s^{\prime}}+s^{\prime}-1}(\mathcal{S})
	\end{pmatrix}.$$
	
\end{lemma}

The following theorem gives necessary and sufficient conditions for these codes to be MDS.

\begin{theorem}\label{Thm:4.2}
	The code $\operatorname{RCTRS}_{\ell,\eta}(\mathcal{A},\boldsymbol{B},\boldsymbol{\lambda})$ is MDS if and only if the following conditions hold:
	\begin{description}
		\item[(i)] For any $k$-subset $\mathcal{L}\subseteq [n-2]$, we have $$1+(-1)^{k-\ell-1}\eta S_{k-\ell}(\boldsymbol{\alpha}_{\mathcal{L}})\neq 0,$$
        where $S_{k-\ell}(\boldsymbol{\alpha}_{\mathcal{L}})=\sum\limits_{L_{1}\subseteq \mathcal{L}\atop \left|L_{1}\right|=k-\ell}\prod\limits_{i\in L_{1}}\alpha_{i}$.
		\item[(ii)] For any $k-1$-subset $\mathcal{I}\subseteq [n-2]$, we have 
$$\Psi_{\mathcal{I},\ell,\eta}^{(1)}(b_{t})\neq \lambda_{t}\Psi^{(1)}_{\mathcal{I},\ell,\eta}(b_{3}),t=1,2,$$
where $\Psi^{(1)}_{\mathcal{I},\ell,\eta}(x)=\prod\limits_{i\in \mathcal{I}}(x-\alpha_{i})\left(1+(-1)^{k-\ell-1}\eta S_{k-\ell}(\alpha_{\mathcal{I}},x)\right)$.
\item[(iii)] For any $k-2$-subset $\mathcal{J}\subseteq [n-2]$, we have
$$\Psi_{\mathcal{J},\ell,\eta}^{(2)}(b_{1},b_{2})-\lambda_{2}\Psi_{\mathcal{J},\ell,\eta}^{(2)}(b_{1},b_{3})+\lambda_{1}\Psi_{\mathcal{J},\ell,\eta}^{(2)}(b_{2},b_{3})\neq 0,$$
       where $\Psi_{\mathcal{J},\ell,\eta}^{(2)}(x,y)=(y-x)\prod\limits_{1\leq j\leq k-2}(y-\alpha_{i_{j}})(x-\alpha_{i_{j}})\cdot\left(1+(-1)^{k-\ell-1}\eta S_{k-\ell}(\alpha_{\mathcal{J}},x,y)\right)$. 
	\end{description}
\end{theorem}

\begin{proof}
	We need to show that for any subsets $\mathcal{L}=\left\{i_{1},\cdots,i_{k}\right\}$, $\mathcal{I}=\left\{i_{1},\cdots,i_{k-1}\right\}$ and $\mathcal{J}=\left\{i_{1},\cdots,i_{k-2}\right\}$ of $[n-2]$, the following matrices
	\begin{equation*}
	A=\small{\begin{pmatrix}
		1&\cdots&1\\
		\alpha_{i_{1}}&\cdots&\alpha_{i_{k}}\\
		\vdots&\vdots&\vdots\\
		\alpha_{i_1}^{\ell-1}&\cdots&\alpha_{i_{k}}^{\ell-1}\\
		\alpha_{i_1}^{\ell+1}&\cdots&\alpha_{i_{k}}^{\ell+1}\\
		\vdots&\vdots&\vdots\\
		\alpha_{i_1}^{k-1}&\cdots&\alpha_{i_{k}}^{k-1}\\
		\alpha_{i_1}^{\ell}+\eta\alpha_{i_1}^{k}&\cdots&\alpha_{i_{k}}^{\ell}+\eta\alpha_{i_{k}}^k
		\end{pmatrix}},
	\end{equation*}
	\begin{equation*}
	C_{t}=\small{\begin{pmatrix}
		1&\cdots&1&1-\lambda_{t}\\
		\alpha_{i_{1}}&\cdots&\alpha_{i_{k-1}}&b_{t}-\lambda_{t}b_{3}\\
		\vdots&\vdots&\vdots&\vdots\\
		\alpha_{i_1}^{\ell-1}&\cdots&\alpha_{i_{k-1}}^{\ell-1}&b_{t}^{\ell-1}-\lambda_{t}b_{3}^{\ell-1}\\
		\alpha_{i_1}^{\ell+1}&\cdots&\alpha_{i_{k-1}}^{\ell+1}&b_{t}^{\ell+1}-\lambda_{t}b_{3}^{\ell+1}\\
		\vdots&\vdots&\vdots&\vdots\\
		\alpha_{i_1}^{k-1}&\cdots&\alpha_{i_{k-1}}^{k-1}&b_{t}^{k-1}-\lambda_{t}b_{3}^{k-1}\\
		\alpha_{i_1}^{\ell}+\eta\alpha_{i_1}^{k}&\cdots&\alpha_{i_{k-1}}^{\ell}+\eta\alpha_{i_{k-1}}^k&b_{t}^{\ell}-\lambda_{t}b_{3}^{\ell}+\eta\left( b_{t}^{k}-\lambda_{t}b_{3}^{k}\right)\\
		\end{pmatrix}},t=1,2
	\end{equation*}
	and 
	\begin{equation*}
     D=\small{\begin{pmatrix}
     	1&\cdots&1&1-\lambda_{1}&1-\lambda_{2}\\
     	\alpha_{i_{1}}&\cdots&\alpha_{i_{k-2}}&b_{1}-\lambda_{1}b_{3}&b_{2}-\lambda_{2}b_{3}\\
     	\vdots&\vdots&\vdots&\vdots&\vdots\\
     	\alpha_{i_{1}}^{\ell-1}&\cdots&\alpha_{i_{k-2}}^{\ell-1}&b_{1}^{\ell-1}-\lambda_{1}b_{3}^{\ell-1}&b_{2}^{\ell-1}-\lambda_{2}b_{3}^{\ell-1}\\
     	\alpha_{i_{1}}^{\ell+1}&\cdots&\alpha_{i_{k-2}}^{\ell+1}&b_{1}^{\ell+1}-\lambda_{1}b_{3}^{\ell+1}&b_{2}^{\ell+1}-\lambda_{2}b_{3}^{\ell+1}\\
     	\vdots&\vdots&\vdots&\vdots&\vdots\\
     	\alpha_{i_{1}}^{k-1}&\cdots&\alpha_{i_{k-2}}^{k-1}&b_{1}^{k-1}-\lambda_{1}b_{3}^{k-1}&b_{2}^{k-1}-\lambda_{2}b_{3}^{k-1}\\
     	\alpha_{i_1}^{\ell}+\eta\alpha_{i_1}^{k}&\cdots&\alpha_{i_{k-2}}^{\ell}+\eta\alpha_{i_{k-2}}^k&b_{1}^{\ell}-\lambda_{1}b_{3}^{\ell}+\eta\left( b_{1}^{k}-\lambda_{1}b_{3}^{k}\right)&b_{2}^{\ell}-\lambda_{2}b_{3}^{\ell}+\eta\left( b_{2}^{k}-\lambda_{2}b_{3}^{k}\right)\\
     	\end{pmatrix}}
	\end{equation*}
	are all invertible. 
    
    (i) Computation of $\det (A)$. From Lemma~\ref{Lem:4.1}, we know that
	\begin{equation*}
	\begin{aligned}
	\det(A)&=(-1)^{k-\ell-1}V(\alpha_{i_{1}},\cdots,\alpha_{i_{k}})+\eta S_{k-\ell}(\alpha_{i_{1}},\cdots,\alpha_{i_{k}})\cdot V(\alpha_{i_{1}},\cdots,\alpha_{i_{k}})\\
	&=(-1)^{k-\ell-1}V(\alpha_{i_{1}},\cdots,\alpha_{i_{k}})(1+(-1)^{k-\ell-1}\eta S_{k-\ell}(\alpha_{i_{1}},\cdots,\alpha_{i_{k}})).
	\end{aligned}
	\end{equation*}
	Thus, $\det(A)\neq 0$ if and only if $1+(-1)^{k-\ell-1}\eta S_{k-\ell}(\alpha_{i_{1}},\cdots,\alpha_{i_{k}})\neq 0$. 
    
    (ii) Computation of $\det (C_{t})$ for $t=1,2$. It is easy to see that \begin{equation*}
	\det(C_{t})=(-1)^{k-\ell-1}\Delta_{1}(b_{t})+(-1)^{k-\ell}\lambda_{t}\Delta_{1}(b_{3})+\eta\Delta_{2}(b_{t})-\lambda_{t}\eta\Delta_{2}(b_{3}),
	\end{equation*}
	where $\Delta_{1}(x)=V(\alpha_{i_{1}},\cdots,\alpha_{i_{k-1}},x)$ and 
	$$\Delta_{2}(x)=\left|\begin{array}{cccc}
	1&\cdots&1&1\\
	\alpha_{i_{1}}&\cdots&\alpha_{i_{k-1}}&x\\
	\vdots&\vdots&\vdots&\vdots\\
	\alpha_{i_{1}}^{\ell-1}&\cdots&\alpha_{i_{k-1}}^{\ell-1}&x^{\ell-1}\\
	\alpha_{i_{1}}^{\ell+1}&\cdots&\alpha_{i_{k-1}}^{\ell+1}&x^{\ell+1}\\
	\vdots&\vdots&\vdots&\vdots\\
	\alpha_{i_{1}}^{k}&\cdots&\alpha_{i_{k-1}}^{k}&x^{k}
	\end{array}\right|=S_{k-\ell}(\alpha_{i_{1}},\cdots,\alpha_{i_{k-1}},x)\cdot V(\alpha_{i_{1}},\cdots,\alpha_{i_{k-1}},x).$$
	 Thus, \small{\begin{equation*}
	\begin{aligned}
	&\det(C_{t})\\
    =&(-1)^{k-\ell-1}V(\alpha_{i_{1}},\cdots,\alpha_{i_{k-1}},b_{t})+\eta S_{k-\ell}(\alpha_{i_{1}},\cdots,\alpha_{i_{k-1}},b_{t})\cdot V(\alpha_{i_{1}},\cdots,\alpha_{i_{k-1}},b_{t})\\
	&+(-1)^{k-\ell}\lambda_{t}V(\alpha_{i_{1}},\cdots,\alpha_{i_{k-1}},b_{3})-\lambda_{t}\eta S_{k-\ell}(\alpha_{i_{1}},\cdots,\alpha_{i_{k-1}},b_{3})\cdot V(\alpha_{i_{1}},\cdots,\alpha_{i_{k-1}},b_{3})\\
	=&(-1)^{k-\ell-1}V(\alpha_{i_{1}},\cdots,\alpha_{i_{k-1}})\left(\prod\limits_{1\leq j\leq k-1}(b_{t}-\alpha_{i_{j}})+(-1)^{k-\ell-1}\eta S_{k-\ell}(\alpha_{i_{1}},\cdots,\alpha_{i_{k-1}},b_{t})\prod\limits_{1\leq j\leq k-1}(b_{t}-\alpha_{i_{j}})\right)\\
	&+(-1)^{k-\ell}\lambda_{t}V(\alpha_{i_{1}},\cdots,\alpha_{i_{k-1}})\left(\prod\limits_{1\leq j\leq k-1}(b_{3}-\alpha_{i_{j}})+(-1)^{k-\ell-1}\eta S_{k-\ell}(\alpha_{i_{1}},\cdots,\alpha_{i_{k-1}},b_{3})\prod\limits_{1\leq j\leq k-1}(b_{3}-\alpha_{i_{j}})\right)\\
	=&(-1)^{k-\ell-1}V(\alpha_{i_{1}},\cdots,\alpha_{i_{k-1}})\left(\Psi^{(1)}_{\mathcal{I},\ell,\eta}(b_{t})-\lambda_{t}\Psi^{(1)}_{\mathcal{I},\ell,\eta}(b_{3})\right),
	\end{aligned}	
	\end{equation*}}
	where $\Psi^{(1)}_{\mathcal{I},\ell,\eta}(x)=\prod\limits_{i\in \mathcal{I}}(x-\alpha_{i})\left(1+(-1)^{k-\ell-1}\eta S_{k-\ell}(\alpha_{\mathcal{I}},x)\right)$. Thus, $\det(C_{t})\neq 0$ if and only if $\Psi_{\mathcal{I},\ell,\eta}^{(1)}(b_{t})\neq \lambda_{t}\Psi^{(1)}_{\mathcal{I},\ell,\eta}(b_{3})$. 
    
    (iii) Computation of $\det (D)$. We have 
    \begin{equation*}
        \begin{aligned}
            \det(D)&=(-1)^{k-\ell-1}\Delta_{3}(b_{1},b_{2})+(-1)^{k-\ell}\lambda_{2}\Delta_{3}(b_{1},b_{3})+(-1)^{k-\ell-1}\lambda_{1}\Delta_{3}(b_{2},b_{3})\\
            &+\eta\Delta_{4}(b_{1},b_{2})-\lambda_{2}\eta\Delta_{4}(b_{1},b_{3})+\eta\lambda_{1}\Delta_{4}(b_{2},b_{3}),
        \end{aligned}
    \end{equation*}
    where $\Delta_{3}(x,y)=V(\alpha_{i_{1}},\cdots,\alpha_{i_{k-2}},x,y)=(y-x)V(\alpha_{i_{1}},\cdots,\alpha_{i_{k-2}})\prod\limits_{1\leq j\leq k-2}(x-\alpha_{i_{j}})(y-\alpha_{i_{j}})$ and 
    $$\Delta_{4}(x,y)=\left|
    \begin{array}{ccccc}
    1&\cdots&1&1&1\\
	\alpha_{i_{1}}&\cdots&\alpha_{i_{k-2}}&x&y\\
	\vdots&\vdots&\vdots&\vdots&\vdots\\
	\alpha_{i_{1}}^{\ell-1}&\cdots&\alpha_{i_{k-2}}^{\ell-1}&x^{\ell-1}&y^{\ell-1}\\
	\alpha_{i_{1}}^{\ell+1}&\cdots&\alpha_{i_{k-2}}^{\ell+1}&x^{\ell+1}&y^{\ell+1}\\
	\vdots&\vdots&\vdots&\vdots&\vdots\\
	\alpha_{i_{1}}^{k}&\cdots&\alpha_{i_{k-2}}^{k}&x^{k}&y^k  
    \end{array}\right|=S_{k-\ell}(\alpha_{i_{1}},\cdots,\alpha_{i_{k-2}},x,y)\cdot V(\alpha_{i_{1}},\cdots,\alpha_{i_{k-2}},x,y).$$
    Thus, \begin{equation*}
        \begin{aligned}
            &\det(D)\\
            =&(-1)^{k-\ell-1}\left(\Delta_{3}(b_{1},b_{2})+(-1)^{k-\ell-1}\eta\Delta_{4}(b_{1},b_{2})\right)+(-1)^{k-\ell}\lambda_{2}\left(\Delta_{3}(b_{1},b_{3})+(-1)^{k-\ell-1}\eta\Delta_{4}(b_{1},b_{3})\right)\\
            &+(-1)^{k-\ell-1}\lambda_{1}\left(\Delta_{3}(b_{2},b_{3})+(-1)^{k-\ell-1}\eta\Delta_{4}(b_{2},b_{3})\right)\\
            =&(-1)^{k-\ell-1}V(\alpha_{i_{1}},\cdots,\alpha_{i_{k-2}})(b_{2}-b_{1})\prod\limits_{1\leq j\leq k-2}(b_{2}-\alpha_{i_{j}})(b_{1}-\alpha_{i_{j}})\cdot\left(1+(-1)^{k-\ell-1}\eta S_{k-\ell}(\alpha_{\mathcal{J}},b_{1},b_{2})\right)\\
            &+(-1)^{k-\ell}\lambda_{2}V(\alpha_{i_{1}},\cdots,\alpha_{i_{k-2}})(b_{3}-b_{1})\prod\limits_{1\leq j\leq k-2}(b_{3}-\alpha_{i_{j}})(b_{1}-\alpha_{i_{j}})\cdot\left(1+(-1)^{k-\ell-1}\eta S_{k-\ell}(\alpha_{\mathcal{J}},b_{1},b_{3})\right)\\
            &+(-1)^{k-\ell-1}\lambda_{1}V(\alpha_{i_{1}},\cdots,\alpha_{i_{k-2}})(b_{3}-b_{2})\prod\limits_{1\leq j\leq k-2}(b_{3}-\alpha_{i_{j}})(b_{2}-\alpha_{i_{j}})\cdot\left(1+(-1)^{k-\ell-1}\eta S_{k-\ell}(\alpha_{\mathcal{J}},b_{2},b_{3})\right)\\
            =&(-1)^{k-\ell-1}V(\alpha_{i_{1}},\cdots,\alpha_{i_{k-2}})\left(\Psi_{\mathcal{J},\ell,\eta}^{(2)}(b_{1},b_{2})-\lambda_{2}\Psi_{\mathcal{J},\ell,\eta}^{(2)}(b_{1},b_{3})+\lambda_{1}\Psi_{\mathcal{J},\ell,\eta}^{(2)}(b_{2},b_{3})\right),
        \end{aligned}
    \end{equation*}
	where $\Psi_{\mathcal{J},\ell,\eta}^{(2)}(x,y)=(y-x)\prod\limits_{1\leq j\leq k-2}(y-\alpha_{i_{j}})(x-\alpha_{i_{j}})\cdot\left(1+(-1)^{k-\ell-1}\eta S_{k-\ell}(\alpha_{\mathcal{J}},x,y)\right)$. Thus, $\det(D)\neq 0$ if and only if $\Psi_{\mathcal{J},\ell,\eta}^{(2)}(b_{1},b_{2})-\lambda_{2}\Psi_{\mathcal{J},\ell,\eta}^{(2)}(b_{1},b_{3})+\lambda_{1}\Psi_{\mathcal{J},\ell,\eta}^{(2)}(b_{2},b_{3})\neq 0$.  This completes the proof.
\end{proof}

We provide a toy example to illustrate the validity of the theorem.
\begin{example}
    Let $n=11,k=7$ and $\mathcal{A}=\{0,1,2,3,4,5,6,7,8\}\subseteq\mathbb{F}_{47}$. Let $b_{1}=9,b_2=10,b_3=11,\lambda_{1}=1,\lambda_{2}=10,\eta=19,\ell=6$.
    Then we have
    \begin{equation}
    \addtocounter{MaxMatrixCols}{11}
G_{\operatorname{RCTRS}}=\begin{pmatrix}
	1&1&1&1&1&1&1&1&1&0&38\\
0&1&2&3&4&5&6&7&8&45&41\\
0&1&4&9&16&25&36&2&17&7&18\\
0&1&8&27&17&31&28&14&42&9&4\\
0&1&16&34&21&14&27&4&7&4&31\\
0&1&32&8&37&23&21&28&9&35&23\\
0&20&5&29&22&42&14&38&18&28&28
	\end{pmatrix}.
\end{equation}
On the one hand, one can check that the above construction satisfies the three conditions of Theorem~\ref{Thm:4.2} using Magma. On the other hand, from the given generator matrix, it can be directly verified  that this code is an MDS code with parameters $[11,7,5]$ over $\mathbb{F}_{47}$ by using Magma. %Therefore, this example verifies the correctness of Theorem~\ref{Thm:4.2}.
    
\end{example}

Next, we present some explicit constructions for the code $\operatorname{RCTRS}_{\ell,\eta}(\mathcal{A},\boldsymbol{B},\boldsymbol{\lambda})$, which are MDS and non-RS MDS codes.

\begin{theorem}\label{The:4.3}
    Let $\mathbb{F}_{q_{0}}\subsetneq\mathbb{F}_{q_{1}}\subsetneq\mathbb{F}_{q_{2}}\subsetneq\mathbb{F}_{q}$ be a chain of subfields of $\mathbb{F}_{q}$. Let $\mathcal{A}=\left\{\alpha_{1},\cdots,\alpha_{n-2}\right\}\subseteq \mathbb{F}_{q_{0}}$ and $\boldsymbol{B}=\left(b_{1},b_{2},b_{3}\right)\in\mathbb{F}_{q_{0}}^3$, where $b_{1},b_{2},b_{3}$ are not all the same, and $\alpha_{i}\neq b_{j}$ for all $1\leq i\leq n-2,1\leq j\leq 3$. Let $\eta\in\mathbb{F}_{q_{1}}^{*}\backslash\mathbb{F}_{q_{0}}^*,0\leq\ell\leq k-1$ and $\boldsymbol{\lambda}=(\lambda_{1},\lambda_{2})\in\mathbb{F}_{q}^2$, where $\lambda_{1}\in \mathbb{F}_{q_{2}}^{*}\backslash \mathbb{F}_{q_{1}}^*,\lambda_{2}\in \mathbb{F}_{q}^{*}\backslash \mathbb{F}_{q_{2}}^*$.
    \begin{description}
        \item[(i)] The code $\operatorname{RCTRS}_{\ell,\eta}(\mathcal{A},\boldsymbol{B},\boldsymbol{\lambda})$ is an MDS code.
        \item[(ii)] If $7\leq k<\frac{n}{2}-1$ and $b_{1},b_{2},b_{3}\in\mathbb{F}_{q}^{*}$ are pairwise distinct, then the dimension of the Schur square code of $\operatorname{RCTRS}_{\ell,\eta}(\mathcal{A},\boldsymbol{B},\boldsymbol{\lambda})$ is $2k+2$ if $\ell\in\{0,k-1\}$ and is $2k+3$ if $1\leq\ell\leq k-2$. Thus, the code $\operatorname{RCTRS}_{\ell,\eta}(\mathcal{A},\boldsymbol{B},\boldsymbol{\lambda})$ is a non-RS MDS code.
    \end{description}
\end{theorem}

\begin{proof}
We first prove that the code $\operatorname{RCTRS}_{\ell,\eta}(\mathcal{A},\boldsymbol{B},\boldsymbol{\lambda})$ is MDS. For any $k$-subset $\mathcal{L}\subseteq [n-2]$, we know that $(-1)^{k-\ell-1}S_{k-\ell}(\boldsymbol{\alpha}_{\mathcal{L}})\in\mathbb{F}_{q_{0}}$
and $\eta\in\mathbb{F}_{q_{1}}^{*}\backslash\mathbb{F}_{q_{0}}^*$. Thus, $1+(-1)^{k-\ell-1}\eta S_{k-\ell}(\boldsymbol{\alpha}_{\mathcal{L}})\neq 0$.
 
 For any $k-1$-subset $\mathcal{I}\subseteq [n-2]$, we have
\[
\frac{\Psi_{\mathcal{I},\ell,\eta}^{(1)}(b_{t})}{\Psi_{\mathcal{I},\ell,\eta}^{(1)}(b_{3})}\in\mathbb{F}_{q_{1}}^{*},\ t=1,2\ \lambda_{1}\in \mathbb{F}_{q_{2}}^{*}\backslash \mathbb{F}_{q_{1}}^*\ \mbox{and}\ \lambda_{2}\in \mathbb{F}_{q}^{*}\backslash \mathbb{F}_{q_{2}}^*
\]
%$$\Psi_{\mathcal{I},\ell,\eta}^{(1)}(b_{t})\in\mathbb{F}_{q_{1}},\ 
%\lambda_{1}\Psi^{(1)}_{\mathcal{I},\ell,\eta}(b_{3})\in\mathbb{F}_{q_{2}}^{*}\backslash \mathbb{F}_{q_{1}}^*,\ \lambda_{2}\Psi^{(1)}_{\mathcal{I},\ell,\eta}(b_{3})\in\mathbb{F}_{q}^{*}\backslash \mathbb{F}_{q_{2}}^*,\  t=1,2.$$
Thus, $$\Psi_{\mathcal{I},\ell,\eta}^{(1)}(b_{t})\neq \lambda_{t}\Psi^{(1)}_{\mathcal{I},\ell,\eta}(b_{3})\  \mbox{for}\ t=1,2.$$
	Next, suppose that there is a $k-2$-subset $\mathcal{J}\subseteq [n-2]$ such that
$$\Psi_{\mathcal{J},\ell,\eta}^{(2)}(b_{1},b_{2})-\lambda_{2}\Psi_{\mathcal{J},\ell,\eta}^{(2)}(b_{1},b_{3})+\lambda_{1}\Psi_{\mathcal{J},\ell,\eta}^{(2)}(b_{2},b_{3})= 0.$$

If $b_{1}\neq b_{3}$, then $\Psi_{\mathcal{J},\ell,\eta}^{(2)}(b_{1},b_{3})\neq 0$. Thus, $$\lambda_{2}=\frac{\Psi_{\mathcal{J},\ell,\eta}^{(2)}(b_{1},b_{2})}{\Psi_{\mathcal{J},\ell,\eta}^{(2)}(b_{1},b_{3})}+\lambda_{1}\frac{\Psi_{\mathcal{J},\ell,\eta}^{(2)}(b_{2},b_{3})}{\Psi_{\mathcal{J},\ell,\eta}^{(2)}(b_{1},b_{3})}.$$
Since $\frac{\Psi_{\mathcal{J},\ell,\eta}^{(2)}(b_{1},b_{2})}{\Psi_{\mathcal{J},\ell,\eta}^{(2)}(b_{1},b_{3})},\frac{\Psi_{\mathcal{J},\ell,\eta}^{(2)}(b_{2},b_{3})}{\Psi_{\mathcal{J},\ell,\eta}^{(2)}(b_{1},b_{3})}\in \mathbb{F}_{q_{1}}$ and $\lambda_{1}\in\mathbb{F}_{q_{2}}^{*}\backslash \mathbb{F}_{q_1}^*$, we know that $\frac{\Psi_{\mathcal{J},\ell,\eta}^{(2)}(b_{1},b_{2})}{\Psi_{\mathcal{J},\ell,\eta}^{(2)}(b_{1},b_{3})}+\lambda_{1}\frac{\Psi_{\mathcal{J},\ell,\eta}^{(2)}(b_{2},b_{3})}{\Psi_{\mathcal{J},\ell,\eta}^{(2)}(b_{1},b_{3})}\in\mathbb{F}_{q_{2}}$, which contradicts $\lambda_{2}\in\mathbb{F}_{q}^{*}\backslash \mathbb{F}_{q_2}^* $

If $b_{1}=b_{3}$, then $b_{2}\neq b_{3},\Psi_{\mathcal{J},\ell,\eta}^{(2)}(b_{1},b_{3})=0$ and $\Psi_{\mathcal{J},\ell,\eta}^{(2)}(b_{2},b_{3})\neq 0$. Thus, 
$$\lambda_{1}=-\frac{\Psi_{\mathcal{J},\ell,\eta}^{(2)}(b_{1},b_{2})}{\Psi_{\mathcal{J},\ell,\eta}^{(2)}(b_{2},b_{3})}.$$
 we know  $-\frac{\Psi_{\mathcal{J},\ell,\eta}^{(2)}(b_{1},b_{2})}{\Psi_{\mathcal{J},\ell,\eta}^{(2)}(b_{2},b_{3})}\in\mathbb{F}_{q_{1}}$, which contradicts $\lambda_{1}\in\mathbb{F}_{q_{2}}^{*}\backslash\mathbb{F}_{q_{1}}^*$.
Thus, for any $k-2$-subset $\mathcal{J}\subseteq [n-2]$, we have
$$\Psi_{\mathcal{J},\ell,\eta}^{(2)}(b_{1},b_{2})-\lambda_{2}\Psi_{\mathcal{J},\ell,\eta}^{(2)}(b_{1},b_{3})+\lambda_{1}\Psi_{\mathcal{J},\ell,\eta}^{(2)}(b_{2},b_{3})\neq 0.$$
From Theorem~\ref{Thm:4.2}, we know that $\operatorname{RCTRS}_{\ell,\eta}(\mathcal{A},\boldsymbol{B},\boldsymbol{\lambda})$ is MDS.  

To prove that $\mathcal{C}:=\operatorname{RCTRS}_{\ell,\eta}(\mathcal{A},\boldsymbol{B},\boldsymbol{\lambda})$ is a  non-RS code, we compute the dimension of the Schur square of $\mathcal{C}$. Let $\boldsymbol{g}_{i}$ be the $i+1$-th row of $G_{\operatorname{RCTRS}}$ and $\boldsymbol{g}_{i,j}=\boldsymbol{g}_{i}\ast \boldsymbol{g}_{j}$ for any $0\leq i,j\leq k-1$.  We divide our discussion into four cases:
    \begin{itemize}
        \item[(I)] If $\ell=0$, choose $(i,j)=(0,2),(1,1)$, then we have
        \begin{equation*}
\begin{aligned}
     \boldsymbol{g}_{1,1}-\boldsymbol{g}_{0,2}&=\left(\alpha_{1}^4,\cdots,\alpha_{n-2}^4,(b_{1}^2-\lambda_{1}b_{3}^2)^2, (b_{2}^2-\lambda_{2}b_{3}^2)^2\right)\\
     &-
     \left(\alpha_{1}^4,\cdots,\alpha_{n-2}^4,(b_{1}-\lambda_{1}b_{3})(b_{1}^3-\lambda_{1}b_{3}^3), (b_{2}-\lambda_{2}b_{3})(b_{2}^3-\lambda_{2}b_{3}^3)\right)\\
     &=(0,\cdots,0,\lambda_{1}b_{1}b_{3}(b_{1}-b_{3})^2,\lambda_{2}b_{2}b_{3}(b_{2}-b_{3})^2)
\end{aligned}
        \end{equation*}
        Choosing $(i,j)=(0,3),(1,2)$, we have
        \begin{equation*}
            \begin{aligned}
                \boldsymbol{g}_{1,2}-\boldsymbol{g}_{0,3}&=\left(\alpha_{1}^5,\cdots,\alpha_{n-2}^5,(b_{1}^2-\lambda_{1}b_{3}^2)(b_{1}^3-\lambda_{1}b_{3}^3), (b_{2}^2-\lambda_{2}b_{3}^2)(b_{2}^3-\lambda_{2}b_{3}^3)\right)\\
     &-
     \left(\alpha_{1}^5,\cdots,\alpha_{n-2}^5,(b_{1}-\lambda_{1}b_{3})(b_{1}^4-\lambda_{1}b_{3}^4), (b_{2}-\lambda_{2}b_{3})(b_{2}^4-\lambda_{2}b_{3}^4)\right)\\
     &=(0,\cdots,0,\lambda_{1}b_{1}b_{3}(b_{1}+b_{3})(b_{1}-b_{3})^2,\lambda_{2}b_{2}b_{3}(b_{2}+b_{3})(b_{2}-b_{3})^2)
            \end{aligned}
        \end{equation*}

	Because $b_{1},b_{2},b_{3}\in\mathbb{F}_{q}^{*}$ are pairwise distinct, we have $$\left| \begin{array}{cc}
	\lambda_{1}b_{1}b_{3}(b_{1}-b_{3})^2&\lambda_{2}b_{2}b_{3}(b_{2}-b_{3})^2\\
	\lambda_{1}b_{1}b_{3}(b_{1}-b_{3})^2(b_{1}+b_{3})&\lambda_{2}b_{2}b_{3}(b_{2}-b_{3})^2(b_{2}+b_{3})
	\end{array}\right|=\lambda_{1}\lambda_{2}b_{1}b_{2}b_{3}^2(b_{1}-b_{3})^2(b_{2}-b_{3})^2(b_{2}-b_{1})\neq 0.$$
    Thus,
    $(0,\cdots,0,\beta_{1},\beta_{2})\in \mathbb{F}_{q}^n$ can be linearly expressed by $\boldsymbol{g}_{1,1}-\boldsymbol{g}_{0,2}$ and $\boldsymbol{g}_{1,2}-\boldsymbol{g}_{0,3}$, where $\beta_{1},\beta_{2}\in \mathbb{F}_{q}$. Thus, for any $(i_{1},j_{1}),(i_{2},j_{2})\in [0,k-2]\times [0,k-2]$ that satisfies $i_{1}+j_{1}=i_{2}+j_{2}$, we get $\boldsymbol{g}_{i_1,j_1}-\boldsymbol{g}_{i_2,j_2}$  that can be linearly expressed by $\boldsymbol{g}_{1,1}-\boldsymbol{g}_{0,2}$ and $\boldsymbol{g}_{1,2}-\boldsymbol{g}_{0,3}$. 
	
	Therefore, $\mathcal{C}^2$ is generated by the following matrix
	\begin{equation*}
	{\footnotesize \begin{pmatrix}
	\alpha_{1}^2&\cdots&\alpha_{n-2}^2&(b_{1}-\lambda_{1}b_{3})^2&(b_{2}-\lambda_{2}b_{3})^2\\
	\alpha_{1}^3&\cdots&\alpha_{n-2}^3&(b_{1}-\lambda_{1}b_{3})(b_1^2-\lambda_{1}b_{3}^2)&(b_{2}-\lambda_{2}b_{3})(b_2^2-\lambda_{2}b_{3}^2)\\
    \vdots&\vdots&\vdots&\vdots&\vdots\\
	\alpha_{1}^k&\cdots&\alpha_{n-2}^k&(b_{1}-\lambda_{1}b_{3})(b_1^{k-1}-\lambda_{1}b_{3}^{k-1})&(b_2-\lambda_{2}b_3)(b_2^{k-1}-\lambda_{2}b_{3}^{k-1})\\
	\alpha_{1}^{k+1}&\cdots&\alpha_{n-2}^{k+1}&(b_{1}^2-\lambda_{1}b_{3}^2)(b_1^{k-1}-\lambda_{1}b_{3}^{k-1})&(b_2^2-\lambda_{2}b_3^2)(b_2^{k-1}-\lambda_{2}b_{3}^{k-1})\\
     \vdots&\vdots&\vdots&\vdots&\vdots\\
	\alpha_{1}^{2k-2}&\cdots&\alpha_{n-2}^{2k-2}&(b_1^{k-1}-\lambda_{1}b_{3}^{k-1})^2&(b_2^{k-1}-\lambda_{2}b_{3}^{k-1})^2\\
	\alpha_{1}^4&\cdots&\alpha_{n-2}^4&(b_{1}^2-\lambda_{1}b_{3}^2)^2&(b_{2}^2-\lambda_{2}b_{3}^2)^2\\
    \alpha_{1}^5&\cdots&\alpha_{n-2}^5&(b_{1}^2-\lambda_{1}b_{3}^2)(b_{1}^3-\lambda_{1}b_{3}^3)&(b_{2}^2-\lambda_{2}b_{3}^2)(b_{2}^3-\lambda_{2}b_{3}^3)\\
	\vdots&\vdots&\vdots&\vdots&\vdots\\
    \alpha_{1}(1+\eta\alpha_{1}^k)&\cdots&\alpha_{n-2}(1+\eta\alpha_{n-2}^k)&(b_{1}-\lambda_{1}b_{3})(1-\lambda_{1}+\eta(b_{1}^k-\lambda_{1}b_{3}^k))&(b_{2}-\lambda_{2}b_{3})(1-\lambda_{2}+\eta(b_{2}^k-\lambda_{2}b_{3}^k))\\
   \alpha_{1}^2(1+\eta\alpha_{1}^k)&\cdots&\alpha_{n-2}^2(1+\eta\alpha_{n-2}^k)&(b_{1}^2-\lambda_{1}b_{3}^2)(1-\lambda_{1}+\eta(b_{1}^k-\lambda_{1}b_{3}^k))&(b_{2}^2-\lambda_{2}b_{3}^2)(1-\lambda_{2}+\eta(b_{2}^k-\lambda_{2}b_{3}^k))\\
	\vdots&\vdots&\vdots&\vdots&\vdots\\
    (1+\eta\alpha_{1}^k)^2&\cdots&(1+\eta\alpha_{n-2}^k)^2&(1-\lambda_{1}+\eta(b_{1}^k-\lambda_{1}b_{3}^k))^2&(1-\lambda_{2}+\eta(b_{2}^k-\lambda_{2}b_{3}^k))^2\\
	\end{pmatrix}},
	\end{equation*}
	which is row equivalent to
	\begin{equation}
	M_1=\begin{pmatrix}
	1+\eta^2 \alpha_{1}^{2k}&\cdots&1+\eta^2\alpha_{n-2}^{2k}&0&0\\
	\alpha_{1}&\cdots&\alpha_{n-2}&0&0\\
	\alpha_{1}^2&\cdots&\alpha_{n-2}^2&0&0\\
    \vdots&\vdots&\vdots&\vdots&\vdots\\
	\alpha_{1}^{2k-1}&\cdots&\alpha_{n-2}^{2k-1}&0&0\\
    0&\cdots&0&\lambda_{1}b_{1}b_{3}(b_{1}-b_{3})^2&\lambda_{2}b_{2}b_{3}(b_{2}-b_{3})^2\\
	0&\cdots&0&\lambda_{1}b_{1}b_{3}(b_{1}-b_{3})^2(b_{1}+b_{3})&\lambda_{2}b_{2}b_{3}(b_{2}-b_{3})^2(b_{2}+b_{3})
	\end{pmatrix}.
	\end{equation}
    Because $n-2>2k$, $\rank\begin{pmatrix}
        1+\eta^2\alpha_{1}^{2k}&\cdots&1+\eta^2\alpha_{n-2}^{2k}\\ \alpha_{1}&\cdots&\alpha_{n-2}\\
        \vdots&\vdots&\vdots\\ \alpha_{1}^{2k-1}&\cdots&\alpha_{n-2}^{2k-1}
    \end{pmatrix}=2k,$
	which means that the rank of matrix $M_1$ is $\rank(M_1)=2k+2$.

\item[(II)] If $1\leq\ell\leq k-3$, since $k\geq 7$,  
$\mathcal{C}^2$ is generated by a matrix equivalent to the following form:
\begin{equation*}
M_2=\begin{pmatrix}
	1&\cdots&1&0&0\\
	\alpha_{1}&\cdots&\alpha_{n-2}&0&0\\
	\alpha_{1}^2&\cdots&\alpha_{n-2}^2&0&0\\
    \vdots&\vdots&\vdots&\vdots&\vdots\\
	\alpha_{1}^{2k}&\cdots&\alpha_{n-2}^{2k}&0&0\\
    0&\cdots&0&\lambda_{1}(b_{1}-b_{3})^2&\lambda_{2}(b_{2}-b_{3})^2\\
	0&\cdots&0&\lambda_{1}(b_{1}-b_{3})^2(b_{1}+b_{3})&\lambda_{2}(b_{2}-b_{3})^2(b_{2}+b_{3})
	\end{pmatrix}.
	\end{equation*}
 Because $n-2>2k$ and $$\left|\begin{array}{cc} \lambda_{1}(b_1-b_3)^2&\lambda_{2}(b_2-b_3)^2\\ \lambda_{1}(b_{1}-b_{3})^2(b_{1}+b_{3})&\lambda_{2}(b_{2}-b_{3})^2(b_{2}+b_{3})\end{array}\right|=\lambda_1\lambda_2(b_1-b_3)^2(b_2-b_3)^2(b_2-b_1)\neq 0,$$
    we know that the rank of matrix $M_2$ is $2k+3$.

\item[(III)] If $\ell=k-2$,   $\mathcal{C}^2$ is generated by a matrix equivalent to the following form:
\begin{equation*}
M_3=\begin{pmatrix}
	1&\cdots&1&0&0\\
	\alpha_{1}&\cdots&\alpha_{n-2}&0&0\\
	\alpha_{1}^2&\cdots&\alpha_{n-2}^2&0&0\\
    \vdots&\vdots&\vdots&\vdots&\vdots\\
	\alpha_{1}^{2k}&\cdots&\alpha_{n-2}^{2k}&0&0\\
    0&\cdots&0&\lambda_{1}(b_{1}-b_{3})^2&\lambda_{2}(b_{2}-b_{3})^2\\
	0&\cdots&0&\lambda_{1}(b_{1}-b_{3})^2(b_{1}+b_{3})&\lambda_{2}(b_{2}-b_{3})^2(b_{2}+b_{3})
	\end{pmatrix}.
	\end{equation*}
    Because $n-2>2k$,
    we know that the rank of matrix $M_3$ is $2k+3$. 

%$$\mathcal{D}=\left\{\boldsymbol{g}_{0,0},\cdots,\boldsymbol{g}_{0,k-2},\boldsymbol{g}_{1,k-2},\cdots,\boldsymbol{g}_{k-2,k-2},\boldsymbol{g}_{1,1},\boldsymbol{g}_{1,2},\boldsymbol{g}_{0,k-1},\boldsymbol{g}_{k-3,k-1},\boldsymbol{g}_{k-2,k-1},\boldsymbol{g}_{k-1,k-1}\right\}$$ has size $2k+3$ and is a basis for $\mathcal{C}^2$.
\item[(IV)] If $\ell=k-1$, $\mathcal{C}^2$ is generated by a matrix equivalent to the following form:
\begin{equation*}
M_4=\begin{pmatrix}
	1&\cdots&1&0&0\\
	\alpha_{1}&\cdots&\alpha_{n-2}&0&0\\
	\alpha_{1}^2&\cdots&\alpha_{n-2}^2&0&0\\
    \vdots&\vdots&\vdots&\vdots&\vdots\\
	\alpha_{1}^{2k-2}&\cdots&\alpha_{n-2}^{2k-2}&0&0\\
    2\eta\alpha_{1}^{2k-1}+\eta^2\alpha_{1}^{2k}&\cdots&2\eta\alpha_{n-2}^{2k-1}+\eta^2\alpha_{n-2}^{2k}&0&0\\
    0&\cdots&0&\lambda_{1}(b_{1}-b_{3})^2&\lambda_{2}(b_{2}-b_{3})^2\\
	0&\cdots&0&\lambda_{1}(b_{1}-b_{3})^2(b_{1}+b_{3})&\lambda_{2}(b_{2}-b_{3})^2(b_{2}+b_{3})
	\end{pmatrix}.
	\end{equation*}
    Since $n-2>2k$,
    we know that the rank of matrix $M_4$ is $2k+2$. 
\end{itemize}

In summary, the dimension of the Schur square code $\mathcal C^2$ is $2k+2$ if $\ell\in\{0,k-1\}$ and is $2k+3$ if $1\leq\ell\leq k-2$. Thus, the code $\mathcal{C}$  is a non-RS MDS code by Proposition~\ref{distinguisher:RS}.
\end{proof}
\begin{remark}
The Schur-square dimension of the constructed row-column twisted family is $2k+2$ or $2k+3$ in the parameter range of Theorem~\ref{The:4.3} $(ii)$, whereas the column-twisted family constructed in Section~3 has Schur square dimension $2k+1$. Therefore, in the cases where these dimensions differ, Schur square dimension separates the present $\operatorname{RCTRS}$ codes from the $\operatorname{CTRS}$ codes constructed in Section~3 and from the corresponding previously known families~\cite{bhagat2025row,liu2025column}. This provides an additional invariant for distinguishing these non-RS MDS codes.
\end{remark}

Indeed, the multiplicative subgroup $\mathbb{F}_{q_{0}}^*$ used in Theorem~\ref{The:4.3}(ii) can be generalized to any multiplicative subgroup of $\mathbb{F}_{q_{1}}^*$. Taking a proper multiplicative subgroup, one can obtain non-RS  MDS codes $\operatorname{CTRS}(\mathcal{A},\boldsymbol{B},\boldsymbol{\lambda})$ with longer lengths compared to those in Theorem~\ref{The:4.3}(ii).

\begin{theorem}
    \label{The:4.6}
    Let $\mathbb{F}_{q_{1}}\subsetneq\mathbb{F}_{q_{2}}\subsetneq\mathbb{F}_{q}$ be a chain of subfields of $\mathbb{F}_{q}$ and $H$ be a subgroup of $\mathbb{F}_{q_{1}}^{*}$ of order $n-1$. Let $H=\left\{1,\mu_{1},\cdots,\mu_{n-2}\right\},\boldsymbol{B}=\left(b_{1},b_{2},b_{3}\right)\in\mathbb{F}_{q_{1}}^3$ and $\boldsymbol{\lambda}=(\lambda_{1},\lambda_{2})$, where 
    $b_{1},b_{2},b_{3}$ are pairwise distinct, $\lambda_{1}\in\mathbb{F}_{q_{1}}^{*}\backslash H,\lambda_{2}\in\mathbb{F}_{q_2}^{*}\backslash\mathbb{F}_{q_{1}}^*$.  Let $\eta\in\mathbb{F}_{q}^{*}\backslash \mathbb{F}_{q_2}^*$ and  $\alpha_{i}=\frac{b_{3}\mu_{i}-b_{1}}{\mu_{i}-1}$ for all $1\leq i\leq n-2$. If $7\leq k<\frac{n}{2}-1$, then 
the code $\operatorname{RCTRS}_{\ell,\eta}(\mathcal{A},\boldsymbol{B},\boldsymbol{\lambda})$ is a non-RS MDS code.
\end{theorem}
\begin{proof}
Since $b_3\mu_i-b_1\neq b_3(\mu_i-1)$ and $b_3\mu_i-b_1\neq b_1(\mu_i-1)$, we know that $\alpha_{i}\neq b_j$ for all $1\leq i\leq n-2$ and $j=1,3$.
 We first prove that the code $\operatorname{RCTRS}_{\ell,\eta}(\mathcal{A},\boldsymbol{B},\boldsymbol{\lambda})$ is MDS. We will divide the proof into three steps.
 \begin{itemize}
     \item [(i)] For any $k$-subset $\mathcal{L}\subseteq [n-2]$,  if $S_{k-\ell}(\boldsymbol{\alpha}_{\mathcal{L}})=0$, then $1+(-1)^{k-\ell-1}\eta S_{k-\ell}(\boldsymbol{\alpha}_{\mathcal{L}})\neq 0$. If $S_{k-\ell}(\boldsymbol{\alpha}_{\mathcal{L}})\neq 0$, then $(-1)^{k-\ell}S_{k-\ell}(\boldsymbol{\alpha}_{\mathcal{L}})^{-1}\in\mathbb{F}_{q_{1}}^{*}$ and $\eta\in\mathbb{F}_{q}^{*}\backslash\mathbb{F}_{q_2}^*$. Thus, $1+(-1)^{k-\ell-1}\eta S_{k-\ell}(\boldsymbol{\alpha}_{\mathcal{L}})\neq 0$. Therefore, for any $k$-subset $\mathcal{L}\subseteq [n-2]$, we have $1+(-1)^{k-\ell-1}\eta S_{k-\ell}(\alpha_{\mathcal{L}})\neq 0.$
     \item [(ii)] Suppose that there exists a $k-1$-subset $\mathcal{I} = \{i_1, i_2, \ldots, i_{k-1}\} \subseteq [n-2]$ and $1\leq t\leq 2$ such that
\[
\Psi_{\mathcal{I},\ell,\eta}^{(1)}(b_t) = \lambda_t \Psi_{\mathcal{I},\ell,\eta}^{(1)}(b_3), 
\]
that is,
\[
\prod_{i \in \mathcal{I}}(b_t - \alpha_i)\left(1 + (-1)^{k-\ell-1}\eta\, S_{k-\ell}(\alpha_\mathcal{I}, b_t)\right) = \lambda_t \prod_{i \in \mathcal{I}}(b_3 - \alpha_i)\left(1 + (-1)^{k-\ell-1}\eta\, S_{k-\ell}(\alpha_\mathcal{I}, b_3)\right).
\]
So
\begin{equation}\label{Eeq:case 2}
 \begin{aligned}
&(-1)^{k-\ell}\eta\left(\prod_{i \in \mathcal I}(b_t - \alpha_i)S_{k-\ell}(\alpha_{\mathcal I}, b_t) - \lambda_t \prod_{i \in \mathcal I}(b_3 - \alpha_i)S_{k-\ell}(\alpha_{\mathcal I}, b_3)\right)\\
&= \prod_{i \in \mathcal I}(b_t - \alpha_i) - \lambda_t \prod_{i \in \mathcal I}(b_3 - \alpha_i). 
\end{aligned}   
\end{equation}

We first show that the right-hand side of Equation~\eqref{Eeq:case 2} is nonzero. For $t = 1$, since $\lambda_1 \in \mathbb{F}_{q_1} \setminus H$ and
\[
\frac{\prod_{i \in \mathcal I}(b_1 - \alpha_i)}{\prod_{i \in \mathcal I}(b_3 - \alpha_i)} = \prod_{i \in I} \mu_i \in H,
\]
we know that the right-hand side of Equation~\eqref{Eeq:case 2} is nonzero. For $t = 2$, since $\lambda_2 \in \mathbb{F}_{q_2}^* \setminus \mathbb{F}_{q_1}^*$ and
\[
\frac{\prod_{i \in I}(b_2 - \alpha_i)}{\prod_{i \in I}(b_3 - \alpha_i)} \in \mathbb{F}_{q_1},
\]
we know that the right-hand side of Equation~\eqref{Eeq:case 2} is also nonzero. Hence, the right-hand side of Equation~\eqref{Eeq:case 2} is nonzero.

Let
\[
a_t = (-1)^{k-\ell}\left(\prod_{i \in\mathcal I}(b_t - \alpha_i)S_{k-\ell}(\alpha_\mathcal{I}, b_t) - \lambda_t \prod_{i \in\mathcal I}(b_3 - \alpha_i)S_{k-\ell}(\alpha_\mathcal{I}, b_3)\right)
\]
for $t=1,2$. If $a_t = 0$, then Equation~\eqref{Eeq:case 2} does not hold. If $a_t \neq 0$, then
\[
\eta = (-1)^{k-\ell}\frac{\prod_{i \in\mathcal I}(b_t - \alpha_i) - \lambda_t \prod_{i \in\mathcal I}(b_3 - \alpha_i)}{a_t} \in \mathbb{F}_{q_2},
\]
which contradicts the assumption that $\eta \in \mathbb{F}_q^* \setminus \mathbb{F}_{q_2}^*$. Thus, for any $k-1$-subset $\mathcal{I}\subseteq [n-2]$, we have  \[
\Psi_{\mathcal{I},\ell,\eta}^{(1)}(b_t) \neq  \lambda_t \Psi_{\mathcal{I},\ell,\eta}^{(1)}(b_3), \quad t = 1, 2.
\]
\item [(iii)] Suppose that there exists a $k-2$-subset $\mathcal J \subseteq [n-2]$ such that
\[
\Psi_{\mathcal J,\ell,\eta}^{(2)}(b_1,b_2) - \lambda_2 \Psi_{\mathcal J,\ell,\eta}^{(2)}(b_1,b_3) + \lambda_1 \Psi_{\mathcal J,\ell,\eta}^{(2)}(b_2,b_3) = 0,
\]
that is,
\begin{align}
&(-1)^{k-\ell}\eta\Bigg((b_2-b_1)\prod_{i \in\mathcal J}(b_2-\alpha_i)(b_1-\alpha_i)S_{k-\ell}(\alpha_{\mathcal J},b_1,b_2) \nonumber\\
&\qquad - \lambda_2(b_3-b_1)\prod_{i \in\mathcal J}(b_3-\alpha_i)(b_1-\alpha_i)S_{k-\ell}(\alpha_{\mathcal J},b_1,b_3) \nonumber\\
&\qquad + \lambda_1(b_3-b_2)\prod_{i \in\mathcal J}(b_2-\alpha_i)(b_3-\alpha_i)S_{k-\ell}(\alpha_{\mathcal J},b_2,b_3)\Bigg) \nonumber\\
&= (b_2-b_1)\prod_{i \in\mathcal J}(b_2-\alpha_i)(b_1-\alpha_i) - \lambda_2(b_3-b_1)\prod_{i \in\mathcal J}(b_3-\alpha_i)(b_1-\alpha_i) \nonumber\\
&\qquad + \lambda_1(b_3-b_2)\prod_{i \in\mathcal J}(b_3-\alpha_i)(b_2-\alpha_i). \label{Eeq:case 3}
\end{align}

Similarly to Step (ii), we first prove that the right-hand side of Equation~\eqref{Eeq:case 3} is nonzero. If the right-hand side of Equation~\eqref{Eeq:case 3} equals zero, then
\[
(b_2-b_1)\prod_{i \in\mathcal J}\frac{b_2-\alpha_i}{b_3-\alpha_i}\mu_i - \lambda_2(b_3-b_1)\prod_{i \in\mathcal J}\mu_i + \lambda_1(b_3-b_2)\prod_{i \in\mathcal J}\frac{b_2-\alpha_i}{b_3-\alpha_i} = 0.
\]
Because $b_1,b_2,b_3$ are pairwise distinct, then
\[
\lambda_2 = \frac{b_2-b_1}{b_3-b_1}\prod_{i \in\mathcal J}\frac{b_2-\alpha_i}{b_3-\alpha_i} + \lambda_1\frac{b_3-b_2}{b_3-b_1}\prod_{i \in \mathcal J}\frac{b_2-\alpha_i}{b_3-\alpha_i}\mu_i^{-1} \in \mathbb{F}_{q_1},
\]
which contradicts the assumption that $\lambda_2 \in \mathbb{F}_{q_2}^* \setminus \mathbb{F}_{q_1}$.
Therefore, the right-hand side of Equation~\eqref{Eeq:case 3} is nonzero. Let
\begin{equation*}
\begin{aligned}
m=&(b_2-b_1)\prod_{i \in\mathcal J}(b_2-\alpha_i)(b_1-\alpha_i)S_{k-\ell}(\alpha_{\mathcal J},b_1,b_2)\\
-& \lambda_2(b_3-b_1)\prod_{i \in\mathcal J}(b_3-\alpha_i)(b_1-\alpha_i)S_{k-\ell}(\alpha_{\mathcal J},b_1,b_3)\\
+& \lambda_1(b_3-b_2)\prod_{i \in\mathcal J}(b_2-\alpha_i)(b_3-\alpha_i)S_{k-\ell}(\alpha_{\mathcal J},b_2,b_3).
\end{aligned}   
\end{equation*}
 If $m=0$, then Equation~\eqref{Eeq:case 3} does not hold. If $m\neq 0$, then 
\begin{equation*}
    \begin{aligned}
        \eta=(-1)^{k-\ell}m^{-1}&\left((b_2-b_1)\prod_{i \in\mathcal J}(b_2-\alpha_i)(b_1-\alpha_i) - \lambda_2(b_3-b_1)\prod_{i \in\mathcal J}(b_3-\alpha_i)(b_1-\alpha_i)\right.\\
&+\left. \lambda_1(b_3-b_2)\prod_{i \in\mathcal J}(b_3-\alpha_i)(b_2-\alpha_i)\right)\in\mathbb{F}_{q_2},
    \end{aligned}
\end{equation*}
which contradicts the assumption that $\eta \in \mathbb{F}_{q}^* \setminus \mathbb{F}_{q_2}^*$. Thus, for any $k-2$-subset $\mathcal J \subseteq [n-2]$, we have
\[
\Psi_{\mathcal J,\ell,\eta}^{(2)}(b_1,b_2) - \lambda_2 \Psi_{\mathcal J,\ell,\eta}^{(2)}(b_1,b_3) + \lambda_1 \Psi_{\mathcal J,\ell,\eta}^{(2)}(b_2,b_3)\neq 0.
\]
 \end{itemize}

 From Theorem~\ref{Thm:4.2}, we know that $\operatorname{RCTRS}_{\ell,\eta}(\mathcal{A},\boldsymbol{B},\boldsymbol{\lambda})$ is MDS. 
    
    For calculating the dimension of the Schur square of $\operatorname{RCTRS}_{\ell,\eta}(\mathcal{A},\boldsymbol{B},\boldsymbol{\lambda})$,  similar to Theorem~\ref{The:4.3}, the dimension of the Schur square code of $\operatorname{RCTRS}_{\ell,\eta}(\mathcal{A},\boldsymbol{B},\boldsymbol{\lambda})$ is at least $2k+2$ for $7\leq k<\frac{n}{2}-1$. Thus, the code $\operatorname{RCTRS}_{\ell,\eta}(\mathcal{A},\boldsymbol{B},\boldsymbol{\lambda})$ is a non-RS MDS code by Proposition~\ref{distinguisher:RS}.

 \end{proof}

\section{The dual codes of \texorpdfstring{$\operatorname{CTRS}(\mathcal{A},\boldsymbol{B},\boldsymbol{\lambda})$ and $\operatorname{RCTRS}_{\ell,\eta}(\mathcal{A},\boldsymbol{B},\boldsymbol{\lambda})$}{CTRS(A,B,lambda) and RCTRS(l,eta,A,B,lambda)}}\label{sec5}
In this section, we determine the dual codes or parity-check matrices of $\operatorname{CTRS}(\mathcal{A},\boldsymbol{B},\boldsymbol{\lambda})$ and $\operatorname{RCTRS}_{\ell,\eta}(\mathcal{A},\boldsymbol{B},\boldsymbol{\lambda})$, respectively.

For any two vectors $\boldsymbol{x}=(x_{1},\cdots,x_{n}),\boldsymbol{y}=(y_{1},\cdots,y_{n})\in\mathbb{F}_{q}^n$, the inner product of $\boldsymbol{x}$ and $\boldsymbol{y}$ is defined as $\boldsymbol{x}\cdot\boldsymbol{y}=\sum\limits_{i=1}^{n}x_{i}y_{i}$. The dual code of a linear code $\mathcal{C}$ is defined to be
$$
\mathcal{C}^{\perp}=\left\{\boldsymbol{x} \in \mathbb{F}_{q}^{n}: \boldsymbol{x} \cdot \boldsymbol{c}=\sum_{i=1}^{n} c_{i} x_{i}=0 \mbox { for all } \boldsymbol{c} \in \mathcal{C}\right\}.
$$
A generator matrix of the dual code $\mathcal{C}^{\perp}$ is called a parity-check matrix of the code $\mathcal{C}$.

Firstly, we recall a useful result from~\cite{sui2022mds1}.
\begin{lemma}[\cite{sui2022mds1}]\label{lemcont:4.1}
	Let $\alpha_{1},\cdots,\alpha_{n-2}$ be distinct elements of $\mathbb{F}_{q}$ and $\prod\limits_{i=1}^{n-2}(x-\alpha_{i})=\sum\limits_{j=0}^{n-2}\sigma_{j}x^{n-2-j}$. Let $\Lambda_{0}=1$ and $\boldsymbol{y}=(\Lambda_{0},\Lambda_{1},\cdots,\Lambda_{n-2})$ be the unique solution of the following system of equations:
	\begin{equation*}
	\begin{pmatrix}
	\sigma_{0}&0&0&\cdots&0\\
	\sigma_{1}&\sigma_{0}&0&\cdots&0\\
	\sigma_{2}&\sigma_{1}&\sigma_{0}&\cdots&0\\
	\vdots&\vdots&\vdots&\ddots&\vdots\\
	\sigma_{n-2}&\sigma_{n-3}&\sigma_{n-4}&\cdots&\sigma_{0}
	\end{pmatrix}
	\begin{pmatrix}
	\Lambda_{0}\\\Lambda_{1}\\ \vdots\\ \Lambda_{n-2}
	\end{pmatrix}=
	\begin{pmatrix}
	1\\0\\ \vdots\\ 0
	\end{pmatrix}.
	\end{equation*}
	For any fixed $0\leq t\leq n-2$, if $\alpha_{i}^{n-3+t}=\sum\limits_{j=0}^{n-3}f_{j}\alpha_{i}^j$ for $1\leq i\leq n-2$, then $f_{n-3}=\Lambda_{t}$.
\end{lemma}

{We first give a parity-check matrix for the $\operatorname{RCTRS}$ code. Then as a special sub-class (setting $\eta=0$), we can obtain  a parity-check matrix for the $\operatorname{CTRS}$ code.}
\begin{theorem}\label{The:5.2}
	Let $3\leq k\leq n-3$. Let $\mathcal{A}=\left\{\alpha_{1},\cdots,\alpha_{n-2}\right\}\subseteq\mathbb{F}_{q},\prod\limits_{i=1}^{n-2}(x-\alpha_{i})=\sum\limits_{j=0}^{n-2}\sigma_{j}x^{n-2-j}$ and $u_{i}=\prod\limits_{1\leq j\leq n-2\atop j\neq i}(\alpha_{i}-\alpha_{j})^{-1}$ for all $1\leq i\leq n-2$. Let $\boldsymbol{B}=\left(b_{1},b_{2},b_{3}\right)\in\mathbb{F}_{q}^3,\boldsymbol{\lambda}=\left(\lambda_{1},\lambda_{2}\right)\in\mathbb{F}_{q}^2$ and $\eta\in\mathbb{F}_{q},0\leq\ell\leq k-1$. Then $\mathcal{C}=\operatorname{RCTRS}_{\ell,\eta}(\mathcal{A},\boldsymbol{B},\boldsymbol{\lambda})$ has a parity-check matrix of the following form 
    \begin{equation}\label{equ:4.1}
	H=\begin{pmatrix}
	\cdots&u_{r}&\cdots&0&0\\
    \cdots&u_{r}\alpha_{r}&\cdots&0&0\\
    \vdots&\vdots&\vdots&\vdots&\vdots\\
    \cdots&u_{r}\alpha_{r}^{n-k-4}&\cdots&0&0\\
    \cdots&u_{r}f(\alpha_{r})&\cdots&0&0\\
    \cdots&u_{r}\sum\limits_{i=n-k-3}^{n-3}\sum\limits_{j=0}^{n-3-i}\sigma_{j}\left(b_{1}^{n-3-i-j}-\lambda_{1}b_{3}^{n-3-i-j}\right)\alpha_{r}^i&\cdots&-1&0\\
    \cdots&u_{r}\sum\limits_{i=n-k-3}^{n-3}\sum\limits_{j=0}^{n-3-i}\sigma_{j}\left(b_{2}^{n-3-i-j}-\lambda_{2}b_{3}^{n-3-i-j}\right)\alpha_{r}^i&\cdots&0&-1\\
	\end{pmatrix}_{1\leq r\leq n-2}\in\mathbb{F}_{q}^{(n-k)\times n},
	\end{equation}
    where $f(x)=x^{n-k-3}\left(1-\eta\sum\limits_{j=0}^{k-\ell}\sigma_{j}x^{k-\ell-j}\right)$ and $\sigma_{j}=(-1)^j\sum\limits_{I\subseteq [n-2]\atop \left|I\right|=j}\prod\limits_{i\in I}\alpha_{i}$.
\end{theorem}

\begin{proof}
    Let $G_{1}\in\mathbb{F}_{q}^{k\times (n-2)}$  be the generator matrix  of $TRS_{k}(\mathcal{A},\ell,\eta)$. From \cite[Theorem 2.3]{gu2025deep}, we know that 
    \begin{equation*}
        H_{1}=\begin{pmatrix}
            u_{1}&\cdots&u_{n-2}\\
            u_{1}\alpha_{1}&\cdots&u_{n-2}\alpha_{n-2}\\
            \vdots&\vdots&\vdots\\
            u_{1}\alpha_{1}^{n-k-4}&\cdots&u_{n-2}\alpha_{n-2}^{n-k-4}\\
            u_{1}f(\alpha_{1})&\cdots&u_{n-2}f(\alpha_{n-2})\\
        \end{pmatrix}
    \end{equation*}
    is a parity-check matrix of $TRS_{k}(\mathcal{A},\ell,\eta)$, where $f(x)=x^{n-k-3}\left(1-\eta\sum\limits_{j=0}^{k-\ell}\sigma_{j}x^{k-\ell-j}\right)$. Let $G_{\operatorname{RCTRS}}=\begin{pmatrix}
G_{1}&\boldsymbol{\beta}_{1}&\boldsymbol{\beta}_{2}
    \end{pmatrix},$ 
    where for $t=1,2$,
    $$\boldsymbol{\beta}_{t}^T=\left(1-\lambda_{t},\cdots,b_{t}^{\ell-1}-\lambda_{t}b_{3}^{\ell-1},b_{t}^{\ell+1}-\lambda_{t}b_{3}^{\ell+1},\cdots,b_{t}^{k-1}-\lambda_{t}b_{3}^{k-1},b_{t}^{\ell}-\lambda_{t}b_{3}^{\ell}+\eta(b_{t}^{k}-\lambda_{t}b_{3}^{k})\right)\in\mathbb{F}_{q}^{k}.$$
    Obviously, the code $\operatorname{RCTRS}_{\ell,\eta}(\mathcal{A},\boldsymbol{B},\boldsymbol{\lambda})$ has a parity-check matrix of the following form:
    \[ H_{\operatorname{RCTRS}}=\begin{pmatrix}
        H_{1}&\boldsymbol{0}_{(n-k-2)\times 1}&\boldsymbol{0}_{(n-k-2)\times 1}\\
        \boldsymbol{\gamma}_{1}& -1&0\\
        \boldsymbol{\gamma}_{2}&0&-1
    \end{pmatrix},\]
    where $\gamma_{1},\gamma_{2}\in \mathbb{F}_{q}^{ n-2}$.
    It suffices to construct $\gamma_t$ satisfying the stronger system
$G_2\gamma_t^{T}=\phi_t$. Since the rows of $G_1$ are obtained from the
rows of $G_2$ by replacing $x^\ell$ with $x^\ell+\eta x^k$ and deleting
the unused row, this stronger system implies $G_1\gamma_t^{T}=\beta_t$ for $1\leq t\leq 2$. Let $$G_{2}=\begin{pmatrix}
        1&1&\cdots&1\\ \alpha_{1}&\alpha_{2}&\cdots&\alpha_{n-2}\\
        \vdots&\vdots&\vdots&\vdots\\
        \alpha_{1}^{k-1}&\alpha_{2}^{k-1}&\cdots&\alpha_{n-2}^{k-1}\\
        \alpha_{1}^{k}& \alpha_{2}^{k}&\cdots&\alpha_{n-2}^{k}
    \end{pmatrix}$$ and $$\phi_{t}^{T}=\left(1-\lambda_{t},b_{t}-\lambda_{t}b_{3},\cdots,b_{t}^{k-1}-\lambda_{t}b_{3}^{k-1},b_{t}^{k}-\lambda_{t}b_{3}^k\right).$$ 
    Because the solution of equation $G_{2}\boldsymbol{x}^T=\phi_{t}$ must be the solution of $G_{1}\boldsymbol{x}^T=\beta_{t}$, thus it suffices to verify $G_{2}\gamma^T_{t}=\phi_{t}$ for $t=1,2$.
    %{\color{red}Because the solution of equation $G_{2}\gamma_{t}^T=\phi_{t}$ must be the solution of $G_{1}\gamma_{t}^T=\beta_{t}$, thus it suffices to verify $G_{2}\gamma^T_{t}=\phi_{t}$ for $t=1,2$.}

For any nonnegative integer $r$ let $w_{r}=\sum\limits_{i=1}^{n-2}u_{{i}}\alpha_{i}^r$, then 
	\begin{equation}\label{equation:3.1}
	    w_{r}=
    \begin{cases}
       0&\mbox{if}\ 0\leq r\leq n-4\\
	1&\mbox{if}\ r=n-3
    \end{cases}    .
	\end{equation}
%    \left\{
%	\begin{array}{ll}
%	0&if\ 1\leq t\leq n-k-3\\
%	1&if\ t=n-k-2
%	\end{array}\right. 
    Let 
    \begin{equation*}
        \Lambda_{0}=1\ \mbox{and}\ \Lambda_{i}=-\sum\limits_{j=1}^{i}\sigma_{j}\Lambda_{i-j},\ \forall 1\leq i\leq n-2.
    \end{equation*}
     For any fixed $0\leq r\leq n-k-3$, by  Lagrange interpolation there exists $f_{r,0},f_{r,1},\cdots, f_{r,n-3}$ such that
    \[
    \alpha_{i}^{n-3+r}=\sum\limits_{j=0}^{n-3}f_{r,j}\alpha_{i}^j,\,\forall 1\leq i\leq n-2.
    \]
     By Lemma~\ref{lemcont:4.1}, we have $f_{r,n-3}=\Lambda_{r}$. So \begin{equation}\label{equation:3.2}
	w_{n-3+r}=\sum\limits_{i=1}^{n-2}u_{i}\alpha_{i}^{n-3+r}=\sum\limits_{j=0}^{n-3}f_{r,j}\sum\limits_{i=1}^{n-2}u_{i}\alpha_{i}^j=f_{r,n-3}w_{n-3}=f_{r,n-3}=\Lambda_{r}.
	\end{equation}
    Let
    \begin{equation*}
        \gamma_{t}^T=\left(u_{1}\sum\limits_{i=0}^{n-3}a_{i}^{(t)}\alpha_{1}^i,\cdots,u_{n-2}\sum\limits_{i=0}^{n-3}a_{i}^{(t)}\alpha_{n-2}^i\right),
    \end{equation*}
    then
 \begin{equation*}
        \begin{aligned}
            G_{2}\cdot\gamma_{t}^T&=\begin{pmatrix}
                1&1&\cdots&1\\ \alpha_{1}&\alpha_{2}&\cdots&\alpha_{n-2}\\
                \vdots&\vdots&\vdots&\vdots\\
\alpha_{1}^{k}&\alpha_{2}^{k}&\cdots&\alpha_{n-2}^{k}
            \end{pmatrix}\cdot\begin{pmatrix}
                   u_{1}\sum\limits_{i=0}^{n-3}a_{i}^{(t)}\alpha_{1}^i\\
                   \vdots\\
                   u_{n-2}\sum\limits_{i=0}^{n-3}a_{i}^{(t)}\alpha_{n-2}^i
                \end{pmatrix}
                =\begin{pmatrix}
                    a_{n-3}^{(t)}\Lambda_{0}\\
                    a_{n-4}^{(t)}\Lambda_{0}+a_{n-3}^{(t)}\Lambda_{1}\\
                    \vdots\\
                    \sum\limits_{i=n-k-3}^{n-3}a_{i}^{(t)}\Lambda_{i-(n-k-3)}
                \end{pmatrix}\\
                &=\small{\begin{pmatrix}
            \Lambda_{0}&0&\cdots&0\\
    \Lambda_{1}&\Lambda_{0}&\cdots&0\\
    \vdots&\vdots&\vdots&\vdots\\
    \Lambda_{k}&\Lambda_{k-1}&\cdots&\Lambda_{0}
                \end{pmatrix}\cdot\begin{pmatrix}
                    a_{n-3}^{(t)}\\a_{n-4}^{(t)}\\
                    \vdots\\ a_{n-k-3}^{(t)}
                \end{pmatrix}}.
        \end{aligned}
    \end{equation*} 
From Lemma~\ref{lemcont:4.1}, we know that
    \[  \begin{pmatrix}
        \sigma_{0}&0&\cdots&0\\ \sigma_{1}&\sigma_{0}&\cdots&0\\
        \vdots&\vdots&\vdots&\vdots\\
        \sigma_{k}&\sigma_{k-1}&\cdots&\sigma_{0}
    \end{pmatrix}\cdot \begin{pmatrix}
                    \Lambda_{0}&0&\cdots&0\\
                 \Lambda_{1}&\Lambda_{0}&\cdots&0\\
        \vdots&\vdots&\vdots&\vdots\\
                    \Lambda_{k}&\Lambda_{k-1}&\cdots&\Lambda_{0}
                \end{pmatrix}=I_{k+1}.   \]
    Thus, we just need 
    \begin{equation*}
        \begin{aligned}
            \begin{pmatrix}
                    a_{n-3}^{(t)}\\ a_{n-4}^{(t)}\\ \vdots\\
                    a_{n-k-3}^{(t)}
            \end{pmatrix}        
                    &=\begin{pmatrix}
        \sigma_{0}&0&\cdots&0\\ \sigma_{1}&\sigma_{0}&\cdots&0\\
        \vdots&\vdots&\vdots&\vdots\\
        \sigma_{k}&\sigma_{k-1}&\cdots&\sigma_{0}
    \end{pmatrix}\cdot G_{2}\cdot\phi_{t}^T=\begin{pmatrix}
        \sigma_{0}&0&\cdots&0\\ \sigma_{1}&\sigma_{0}&\cdots&0\\
        \vdots&\vdots&\vdots&\vdots\\
        \sigma_{k}&\sigma_{k-1}&\cdots&\sigma_{0}
    \end{pmatrix}\cdot \begin{pmatrix}
        1-\lambda_{t}\\ b_{t}-\lambda_{t}b_{3}\\
        \vdots\\ b_{t}^{k}-\lambda_{t}b_{3}^{k}
    \end{pmatrix}\\
    &=\begin{pmatrix}
        1-\lambda_{t}\\
        \sum\limits_{i=0}^1\sigma_{i}\left(b_{t}^{1-i}-\lambda_{t}b_{3}^{1-i}\right)\\
        \vdots\\
        \sum\limits_{i=0}^{k}\sigma_{i}\left(b_{t}^{k-i}-\lambda_{t}b_{3}^{k-i}\right)
    \end{pmatrix} .
        \end{aligned}
    \end{equation*}
    In other words, 
    \begin{equation*}
        \begin{aligned}
\gamma_{t}^{T}&=\left(\cdots,u_{r}\sum\limits_{i=n-k-3}^{n-3}a_{i}^{(t)}\alpha_{r}^i,\cdots\right)_{1\leq r\leq n-2}\\
            &=\left(\cdots,u_{r}\sum\limits_{i=n-k-3}^{n-3}\sum\limits_{j=0}^{n-3-i}\sigma_{j}\left(b_{t}^{n-3-i-j}-\lambda_{t}b_{3}^{n-3-i-j}\right)\alpha_{r}^i,\cdots\right)_{1\leq r\leq n-2}
        \end{aligned}
    \end{equation*}
       
\end{proof}

Taking $\eta=0$ in Theorem~\ref{The:5.2}, we can obtain a parity-check matrix of the code $\operatorname{CTRS}(\mathcal{A},\boldsymbol{B},\boldsymbol{\lambda})$.

\begin{theorem}\label{Lem:5.3}
Let $3\leq k\leq n-3$.	Let $\mathcal{A}=\left\{\alpha_{1},\cdots,\alpha_{n-2}\right\}\subseteq\mathbb{F}_{q},\prod\limits_{i=1}^{n-2}(x-\alpha_{i})=\sum\limits_{j=0}^{n-2}\sigma_{j}x^{n-2-j}$ and $u_{i}=\prod\limits_{1\leq j\leq n-2\atop j\neq i}(\alpha_{i}-\alpha_{j})^{-1}$ for all $1\leq i\leq n-2$. Let $\boldsymbol{B}=\left(b_{1},b_{2},b_{3}\right)\in\mathbb{F}_{q}^3$ and $\boldsymbol{\lambda}=\left(\lambda_{1},\lambda_{2}\right)\in\mathbb{F}_{q}^2$. Then $\mathcal{C}=\operatorname{CTRS}(\mathcal{A},\boldsymbol{B},\boldsymbol{\lambda})$ has a parity-check matrix of the following form
    \begin{equation}\label{equ:4.1}
	H=\begin{pmatrix}
	\cdots&u_{r}&\cdots&0&0\\
    \cdots&u_{r}\alpha_{r}&\cdots&0&0\\
    \vdots&\vdots&\vdots&\vdots&\vdots\\
    \cdots&u_{r}\alpha_{r}^{n-k-3}&\cdots&0&0\\
    \cdots&u_{r}\sum\limits_{i=n-k-2}^{n-3}\sum\limits_{j=0}^{n-3-i}\sigma_{j}\left(b_{1}^{n-3-i-j}-\lambda_{1}b_{3}^{n-3-i-j}\right)\alpha_{r}^i&\cdots&-1&0\\
    \cdots&u_{r}\sum\limits_{i=n-k-2}^{n-3}\sum\limits_{j=0}^{n-3-i}\sigma_{j}\left(b_{2}^{n-3-i-j}-\lambda_{2}b_{3}^{n-3-i-j}\right)\alpha_{r}^i&\cdots&0&-1\\
	\end{pmatrix}_{1\leq r\leq n-2}.
	\end{equation}
\end{theorem}

\section{Conclusion}\label{sec6}
In this paper, we introduced and studied two variants of twisted Reed-Solomon codes, namely column-twisted Reed-Solomon codes and row-column twisted Reed-Solomon codes. For both families, we provide necessary and sufficient conditions for the MDS property by reducing the relevant k-column minors to explicit subset product and elementary-symmetric-function conditions. These necessary and sufficient conditions lead to explicit field-extension and subgroup-based constructions of MDS codes.

We further studied the Schur squares of the constructed codes. For the column-twisted family, the Schur square has dimension $2k+1$, which distinguishes the resulting MDS codes from generalized Reed-Solomon codes. For the row-column twisted family, the larger Schur-square dimension gives another way to distinguish these codes from $\operatorname{GRS}$ codes and known twisted families.
 Finally, we derived explicit parity-check matrices and dual descriptions for the two families.
  
  For future research, the $\operatorname{RCTRS}$ codes associated with more general $(\mathcal{L},\mathcal{P})$-twist form may be considered. It would also be interesting to investigate their self-orthogonality, self-duality, Euclidean and Hermitian hulls, and applications to quantum-code constructions.
  %That is to say, we can consider  RCTRS codes of the following form: 
  %\[
  %\mathcal{C}=\left\{(f(\alpha_{1}),\cdots,f(\alpha_{n-\ell}),f(b_{1})-\lambda_{1}f(b_{\ell+1}),\cdots,f(b_{\ell})-\lambda_{\ell}f(b_{\ell+1})):f(x)\in F(\mathcal{L},\mathcal{P},B)\right\},
%  \]
  %where $\lambda_{1},\cdots,\lambda_{\ell}\in\mathbb{F}_{q},\alpha_{1},\cdots,\alpha_{n-\ell}\in\mathbb{F}_{q}$ are distinct, $b_{1},\cdots,b_{\ell}$ are also distinct and $F(\mathcal{L},\mathcal{P},B)$  is given by Equation~\eqref{Equ:F,L,P,B}. Furthermore, the MDS property, non-RS property, and duality of the code $\mathcal{C}$ are worth investigating.

\bibliographystyle{plain}
\bibliography{RCTRS}

\end{document}